\documentclass[11pt]{article}

\newenvironment{proof}{\makebox[7ex][l]{\it Proof:\/}}{\hfill\/ \hfill\/ {\it Q.E.D.} \vspace{0.5ex}\\}

\usepackage{xcolor,enumitem,graphicx,epstopdf,caption,subcaption}
\usepackage{mathtools,amsmath,amsgen,amstext,amsbsy,amsopn,amsfonts}
\newtheorem{assumption}{Assumption}
\newtheorem{remark}{Remark}
\newtheorem{theorem}{Theorem}
\newtheorem{lemma}{Lemma}

\newcommand{\bgeq}{\begin{equation}}
\newcommand{\edeq}{\end{equation}}
\newcommand{\bgdm}{\begin{displaymath}}
\newcommand{\eddm}{\end{displaymath}}
\newcommand{\mtx}[1]{\begin{bmatrix} #1 \end{bmatrix}}

\newcommand{\mtxr}[1]{\begin{bmatrix*}[r] #1 \end{bmatrix*}}

\begin{document}

\title{\LARGE \bf
Cooperative Adaptive Learning Control for A Group of Nonholonomic UGVs by Output Feedback
}

\author{Xiaonan Dong$^1$, Paolo Stegagno$^2$,  Chengzhi Yuan$^1$\thanks{Corresponding author.}, Wei Zeng$^3$
\vspace*{0.1in} \\
$^1$Department of Mechanical, Industrial and Systems Engineering \\
University of Rhode Island, Kingston, RI 02881, USA \\
Email: {\tt dong\_xn@uri.edu, cyuan@uri.edu}
\vspace*{0.1in}\\
$^2$Department of Electrical, Computer, and Biomedical Engineering \\
University of Rhode Island, Kingston, RI 02881, USA \\
Email: {\tt pstegagno@uri.edu}
\vspace*{0.1in}\\
$^3$School of Mechanical and Electrical Engineering \\
Longyan University, Longyan 364012, China \\
Email: {\tt zw0597@126.com}
}

\date{}

\maketitle

\begin{abstract}
A high-gain observer-based cooperative deterministic learning (CDL) control algorithm is proposed in this chapter for a group of identical unicycle-type unmanned ground vehicles (UGVs) to track over desired reference trajectories.
For the vehicle states, the positions of the vehicles can be measured, while the velocities are estimated using the high-gain observer.
For the trajectory tracking controller, the radial basis function (RBF) neural network (NN) is used to online estimate the unknown dynamics of the vehicle, and the NN weight convergence and estimation accuracy is guaranteed by CDL.
The major challenge and novelty of this chapter is to track the reference trajectory using this observer-based CDL algorithm without the full knowledge of the vehicle state and vehicle model.
In addition, any vehicle in the system is able to learn the knowledge of unmodeled dynamics along the union of trajectories experienced by all vehicle agents, such that the learned knowledge can be re-used to follow any reference trajectory defined in the learning phase.
The learning-based tracking convergence and consensus learning results, as well as using learned knowledge for tracking experienced trajectories, are shown using the Lyapunov method.
Simulation is given to show the effectiveness of this algorithm.
\end{abstract}


\section{Introduction}

The two-wheel-driven, unicycle-type vehicle is one of the most common mobile robot platforms, and many research results have been published regarding this system \cite{yu2015trajectory, chen2015simple, seyboth2014collective, dong2018cooperative}.
There are two major challenges for controlling this system: the knowledge of all state variables, and the actuate modeling of the system.
For the unicycle-type vehicle that we use in this chapter, the vehicle position and velocity are both required for the trajectory tracking control.
The position of the vehicle can be obtained using cameras or GPS signals, while direct measurement of the vehicle velocity is difficult.
State observer has been proposed to estimate the full state of the system using the measured signals \cite{luenberger1964observing, luenberger1971introduction}, however, traditional observers require the knowledge of the system model for accurate state estimations.
High-gain observer has been proposed to estimate the unmeasured state variables in case that  the system model is not fully known to the observer, and the estimated states can be used for control purposes \cite{lee1997adaptive, khalil2008high, Main, almuatazbellah2018high}.
In this chapter, we follow the standard high-gain observer design method \cite{khalil2008high} to obtain the estimation of vehicle velocity using the measured vehicle position.

For the second challenge, adaptive control has been introduced to deal with system uncertainties \cite{rossomando2015identification, miao2015adaptive}, in which neural network (NN) based control is able to further deal with nonlinear system uncertainties \cite{fierro1998control, rossomando2015identification}.
Though tracking control can be achieved by NN-based adaptive control, however, traditional NN-based control methods failed to achieve parameter (NN weight) convergence.
This shortage requires the controller to update the system parameter (NN weight) all the time when the controller is operating, which is time consuming and computational demanding.
To overcome this deficiency, a deterministic learning (DL) method has been proposed to model the system uncertainties under the partial persistency of excitation (PE) condition \cite{DL}.
To be more specific, it has been shown that the system uncertainties can be accurately modeled with a sufficient large number of radial basis function (RBF) NNs, and local NN weights online updated by DL will converge to their optimal values, provided that the input signal of the RBFNNs is recurrent.

Since the RBFNN estimation is locally accurate around the recurrent trajectory, this becomes a disadvantage when there exists multiple tracking tasks.
The learned knowledge of the system uncertainties, presented by the RBFNNs, cannot be directly applied on a different control task, and it will need a significant amount of storage space for a large number of different tasks.
In recent years, distributed control is a rising topic regarding the control of multiple coordinated agents \cite{cai2015adaptive, yuan2017distributed, yuan2017formation, YuaZD.ISA19, yuan2019cooperative, yuan2019distributed}.
In this chapter, we took the idea of communicating inside the multi-agent system (MAS) and apply it on DL, such that in the learning phase, any vehicle in the MAS is able to learn the unmodeled dynamics not only along its own trajectory, but along the trajectories of all other vehicle agents in this MAS as well.
In other words, the NN weight of any vehicle in this MAS will converge to a common constant, which presents the unmodeled dynamics along the union trajectory of all vehicles, and any vehicle in the MAS is able to use this knowledge to achieve trajectory tracking for any control task learned in the learning phase.

The main contributions of this chapter are summarized as follows.
\begin{enumerate}[label=\roman*)]
	\item 
	A high-gain observer is introduced to estimate the vehicle velocities using the measurement of vehicle position.
	\item 
	An observer and RBFNN-based adaptive learning control algorithm is developed for a multi-vehicle system, such that each vehicle agent will be able to follow the desired reference trajectory.
	\item 
	An online cooperative adaptive NN learning law is proposed, such that the RBFNN weight of all vehicle agents will converge to one common value, which represents the unmodeled dynamics of the vehicle along the union trajectories experienced by all vehicle agents.
	\item 
	An observer and experience-based controller is developed using the common NN model obtained from the learning phase, such that vehicles are able to follow the reference trajectory experienced by any vehicle before with improved control performance.
\end{enumerate}

In the following sections, we briefly describe some preliminaries on graph theory and RBFNNs based DL method, then present the vehicle dynamics and the problem statement, all in section 2.
The main results of this chapter, including the high-gain observer design, CDL-based trajectory tracking control, accurate cooperative learning using RBF NNs, and experience-based trajectory tracking control, are provided in sections 3 and 4, respectively.
Simulation results of an example with four vehicles running three different tasks are provided in section 5.
The conclusions are drawn in section 6.

\textbf{Notations.} $\mathbb{R}$, $\mathbb{R}_+$ and $\mathbb{Z}_+$ denote, respectively, the set of real numbers, the set of positive real numbers and the set of positive integers;
$\mathbb{R}^{m \times n}$ denotes the set of $m \times n$ real matrices; $\mathbb{R}^n$ denotes the set of $n \times 1$ real column vectors; $I_n$ denotes the $n \times n$ identity matrix; $O_{m \times n}$ denotes the zero matrix with dimension of $m \times n$;
Subscript $(\cdot)_{k}$ denotes the $k^{th}$ column vector of a matrix;
$|\cdot|$ is the absolute value of a real number, and $||\cdot||$ is the 2-norm of a vector or a matrix, i.e. $||x||=(x^Tx)^{\frac{1}{2}}$;
$\dot{z}$ denotes the total derivative of $z$ with respect to the time;
$\partial / \partial \mathbf{z}$ denotes the Jacobian matrix as
$\frac{\partial}{\partial \mathbf{z}} = \mtx{\frac{\partial}{\partial z_1} & \cdots & \frac{\partial}{\partial z_n}}$.

\section{Preliminaries and problem statement}

\subsection{Graph theory}
In a graph defined as $\mathcal{G} = (\mathcal{V}, \mathcal{E}, \mathcal{A})$, the elements of $\mathcal{V} = \{1, 2, \dots ,n\}$ are called vertices, the elements of $\mathcal{E}$ are pairs $(i, j)$ with $i, j \in \mathcal{V}, i \neq j$ called edges, and the matrix $\mathcal{A}$ is called the adjacency matrix.
If $(i, j) \in \mathcal{E}$, then agent $i$ is able to receive information from agent $j$, and agent $i$ and $j$ are called adjacent.
The adjacency matrix is thus defined as $\mathcal{A} = [a_{ij}]_{n \times n}$, in which $a_{ij} > 0$ if and only if $(i, j) \in \mathcal{E}$, and $a_{ij} = 0$ otherwise.
For any two nodes $v_i, v_j \in \mathcal{V}$, if there exists a path between them, then the graph $\mathcal{G}$ is called connected.
Furthermore, the graph $\mathcal{G}$ is called fixed if $\mathcal{E}$ and $\mathcal{A}$ do not change over time, and called undirected if $\forall (i, j) \in \mathcal{E}$, pair $(j, i)$ is also in $\mathcal{E}$.
According to \cite{agaev2006matrix}, for the Laplacian matrix $L = [l_{ij}]_{n \times n}$ associated with the undirected graph $\mathcal{G}$, in which
$
	l_{ij} = \begin{cases}
		\sum_{j = 1, j \neq i}^{n} a_{ij} &\quad i = j \\
		-a_{ij} &\quad i \neq j \\
	\end{cases}.
$
If the graph is connected, then $L$ is a positive semi-definite symmetric matrix, with one zero eigenvalue and all other eigenvalues being positive and hence, $\operatorname{rank}(L) \leq n - 1$.

\subsection{Localized RBF Neural Networks and Deterministic Learning}
The RBF networks can be described by $f_{nn}(Z)=\sum_{i=1}^{Nn}w_is_i(Z)=W^TS(Z)$ \cite{park1991universal}, where $Z\in \Omega_Z \subset \mathbb{R}^q$ is the input vector, $W=[w_1,\cdots,w_{N_n}]^T\in \mathbb{R}^{N_n}$ is the weight vector, $N_n$ is the NN node number, and $S(Z)=[s_1(||Z-\mu_1||),\cdots,s_{N_n}(||Z-\mu_{N_n}||)]^T$, with $s_i(\cdot)$ being a radial basis function, and $\mu_i$ $(i=1,2,\cdots,N_n)$ being distinct points in state space.
The Gaussian function $s_i(||Z-\mu_i||)=exp[\frac{-(Z-\mu_i)^T(Z-\mu_i)}{\sigma^2}]$ is one of the most commonly used radial basis functions, where $\mu_i=[\mu_{i1},\mu_{i2},\cdots,\mu_{iq}]^T$ is the center of the receptive field and $\sigma_i$ is the width of the receptive field.
The Gaussian function belongs to the class of localized RBFs in the sense that $s_i(||Z-\mu_i||)\rightarrow 0$ as $||Z||\rightarrow \infty$.
It is easily seen that $S(Z)$ is bounded and there exists a real constant $S_M \in \mathbb{R}_+$ such that $ ||S(Z)||\leq S_M$ \cite{DL}.

It has been shown in \cite{park1991universal,buhmann2003radial} that for any continuous function $f(Z):\Omega_Z \rightarrow \mathbb{R}$ where $\Omega_Z\subset \mathbb{R}^q$ is a compact set, and for the NN approximator, where the node number $N_n$ is sufficiently large, there exists an ideal constant weight vector $W^*$, such that for any $\epsilon^*>0$, $f(Z)=W^{*T}S(Z)+\epsilon,\,\forall Z\in \Omega_Z$, where $|\epsilon|<\epsilon^*$ is the ideal approximation error.
The ideal weight vector $W^*$ is an ``artificial'' quantity required for analysis, and is defined as the value of $W$ that minimizes $|\epsilon|$ for all $Z\in\Omega_Z\subset\mathbb{R}^q$, i.e. $W^*:= \operatorname{arg}\textup{min}_{W\in\mathbb{R}^{N_n}}\{\textup{sup}_{Z\in\Omega_Z}|f(Z)-W^TS(Z)|\}$.
Moreover, based on the localization property of RBF NNs \cite{DL}, for any bounded trajectory $Z(t)$ within the compact set $\Omega_Z$, $f(Z)$ can be approximated by using a limited number of neurons located in a local region along the trajectory: $f(Z)=W^{*T}_\zeta S_\zeta(Z)+\epsilon_\zeta$, where $\epsilon_\zeta$ is the approximation error, with $\epsilon_\zeta=O(\epsilon)=O(\epsilon^*)$, $S_\zeta(Z)=[s_{j1}(Z),\cdots,s_{j\zeta}(Z)]^T\in\mathbb{R}^{N_\zeta}$, $W_\zeta^*=[w^*_{j1},\cdots,w^*_{j\zeta}]^T\in\mathbb{R}^{N_\zeta}$, $N_\zeta<N_n$, and the integers $j_i=j_1,\cdots,j_\zeta$ are defined by $|s_{j_i}(Z_p)|>\theta$ ($\theta>0$ is a small positive constant) for some $Z_p\in Z(k)$.

It is shown in \cite{DL} that for a localized RBF network $W^TS(Z)$ whose centers are placed on a regular lattice, almost any recurrent trajectory
$Z(k)$ (see \cite{DL} for detailed definition of ``recurrent'' trajectories) can lead to the satisfaction of the PE condition of the regressor subvector $S_\zeta(Z)$. This result is recalled in the following lemma.
\begin{lemma}[\cite{DL,wang2006learning}] \label{DL}
Consider any recurrent trajectory $Z(k)$: $\mathbb{Z}_+\rightarrow\mathbb{R}^q$. $Z(k)$ remains in a bounded compact set $\Omega_Z\subset \mathbb{R}^q$, then for RBF network $W^TS(Z)$ with centers placed on a regular lattice (large enough to cover compact set $\Omega_Z$), the regressor subvector $S_\zeta(Z)$ consisting of RBFs with centers located in a small neighborhood of $Z(k)$ is persistently exciting.
\end{lemma}

\subsection{Vehicle model and problem statement}

\begin{figure} [h]
	\centering
		\includegraphics[width=0.35\textwidth]{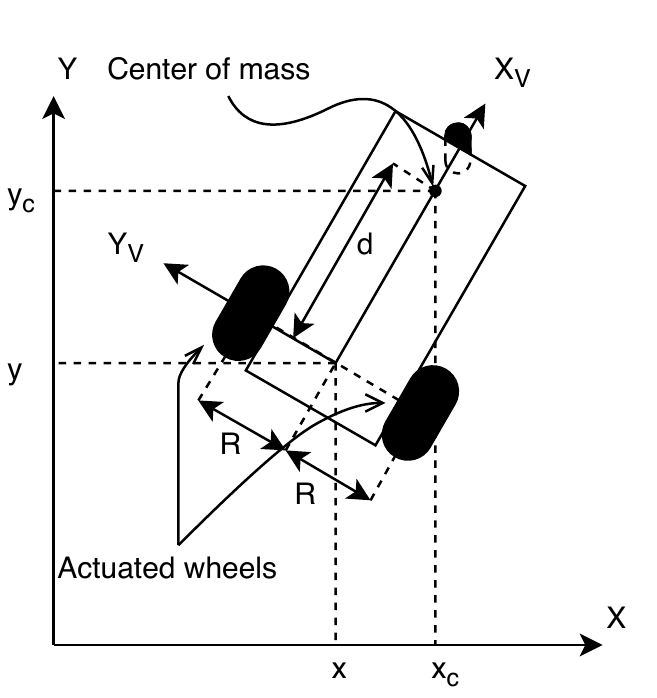}
	\caption{A unicycle-type vehicle}
	\label{unicycle_type_vehicle}
\end{figure}

As shown in Fig. \ref{unicycle_type_vehicle}, this unicycle-type vehicle is a nonholonomic system, with the constraint force preventing the vehicle from sliding along the axis of the actuated wheels.
The nonholonomic constraint can be presented as follows
\begin{equation} \label{constraint}
	A^T(\mathbf{q}_i) \dot{\mathbf q_i} = 0
\end{equation}
in which $A(\mathbf{q}_i) = \mtx{\sin \theta_i & -\cos \theta_i & 0}^T$, and $\mathbf{q}_i = \mtx{x_i & y_i & \theta_i}^T$ is the general coordinates of the $i^\text{th}$ vehicle ($i = 1, 2, \dots, n$, with $n$ being the number of vehicles in the MAS).
($x_i, y_i$) and $\theta_i$ denote the position and orientation of the vehicle with respect to the ground coordinate, respectively.

With this constraint, the degree of freedom of the system is reduced to two.
Independently driven by the two actuated wheels on each side of the vehicle, the non-slippery kinematics of the $i^\text{th}$ vehicle is
\begin{equation} \label{kinematics}
	\dot{\mathbf q}_i
	= \mtx{\dot x_i \\ \dot y_i \\ \dot \theta_i}
	= \mtx{\cos \theta_i & 0 \\ \sin \theta_i & 0 \\ 0 & 1 \\} \mtx{v_i \\ \omega_i}
	\overset{\text{def}}{=} J(\mathbf{q}_i) \mathbf{u}_i
\end{equation}
where $v_i$ and $\omega_i$ are the linear and angular velocities measured at the center between the driving wheels, respectively.
The dynamics of the $i^{th}$ vehicle can be described by \cite{fierro1995control}
\begin{equation} \label{dynamics_q}
	M (\mathbf{q}_i) \ddot{\mathbf{q}}_i + C (\mathbf{q}_i, \dot{\mathbf{q}}_i) \dot{\mathbf{q}_i} + F (\mathbf{q}_i, \dot{\mathbf{q}}_i) + G (\mathbf{q}_i) = B (\mathbf{q}_i) \tau_i + A (\mathbf{q}_i) \lambda_i,
\end{equation}
in which $M \in \mathbb{R}^{3 \times 3}$ is a positive definite matrix that denotes the inertia, $C \in \mathbb{R}^{3 \times 3}$ is the centripetal and Coriolis matrix, $F \in \mathbb{R}^{3 \times 1}$ is the friction vector, $G \in \mathbb{R}^{3 \times 1}$ is the gravity vector.
$\tau_i \in \mathbb{R}^{2 \times 1}$ is a vector of system input, i.e. the torque applied on each driving wheel, $B = \frac{1}{r} \mtx{\cos \theta_i & \cos \theta_i \\ \sin \theta_i & \sin \theta_i \\ R & -R \\} \in \mathbb{R}^{3 \times 2}$ is the input transformation matrix, projecting the system input $\tau$ onto the space spanned by $(x, y, \theta)$, in which $D = 2 R$ is the distance between two actuation wheels, and $r$ is the radius of the wheel.
$\lambda_i$ is a Lagrange multiplier, and $A \lambda_i \in \mathbb{R}^{3 \times 1}$ denotes the constraint force.

Matrices $M$ and $C$ in equation~(\ref{dynamics_q}) can be derived using the Lagrangian equation with the follow steps.
First we calculate the kinetic energy for the $i^\text{th}$ vehicle agent
\begin{equation} \label{kinetic_energy_c}
	T_i = \frac{m (\dot{x}_{ic}^2 + \dot{y}_{ic}^2)}{2} + \frac{I \dot{\theta}_{ic}^2}{2}
\end{equation}
where $m$ is the mass of the vehicle, $I$ is the moment of inertia measured at the center of mass, $x_{ic}$, $y_{ic}$, and $\theta_{ic}$ are the position and orientation of the vehicle at the center of mass, respectively.
The following relation can be obtained from Figure~\ref{unicycle_type_vehicle}:
\begin{equation}
\begin{array}{ll}
		\left\{
			\begin{aligned}
				x_{ic} &= x_i + d \cos \theta_i \\
				y_{ic} &= y_i + d \sin \theta_i \\
				\theta_{ic} &= \theta_i \\
			\end{aligned}
		\right., \quad
		&\left\{
			\begin{aligned}
				\dot{x}_{ic} &= \dot{x}_i - d \dot{\theta} \sin \theta_i \\
				\dot{y}_{ic} &= \dot{y}_i + d \dot{\theta} \cos \theta_i \\
				\dot{\theta}_{ic} &= \dot{\theta}_i \\
			\end{aligned}
		\right.
	\end{array}
\end{equation}
Then equation~(\ref{kinetic_energy_c}) can be rewritten into
\begin{equation} \label{kinetic_energy}
	\begin{split}
		T(\mathbf{q}_i,\dot{\mathbf{q}}_i) &= \frac{m [(\dot{x}_i - d \dot{\theta} \sin \theta_i)^2 + (\dot{y}_i + d \dot{\theta} \cos \theta_i)^2]}{2} + \frac{I \dot{\theta}_i^2}{2} \\
		&= \frac{1}{2} [m \dot{x}_i^2 + m \dot{y}_i^2 + (m d^2 + I) \dot{\theta}^2 - 2 m d \sin \theta \dot{x}_i \dot{\theta}_i + 2 m d \cos \theta \dot{y}_i \dot{\theta}_i ] \\
	 	&= \frac{\dot{\mathbf{q}}_i^T M(\mathbf{q}_i) \dot{\mathbf{q}}_i}{2} \\
	\end{split}
\end{equation}
in which $
	M = \mtx{m & 0 & -m d \sin \theta_i \\
		0 & m & m d \cos \theta_i \\
		-m d \sin \theta_i & m d \cos \theta_i & m d^2 + I \\}
$.
It will be shown later that the inertia matrix $M$ shown above is identical to that in equation~(\ref{dynamics_q}).
Then  the dynamics equation of the system is given by the following Lagrangian equation \cite{book_robotics},
\begin{equation} \label{Lagrangian_eq}
	\frac{\text{d}}{\text{d}t} \left( \frac{\partial L}{\partial \dot{\mathbf{q}}}_i \right)^T - \left( \frac{\partial L}{\partial \mathbf{q}}_i \right)^T = A(\mathbf{q}_i) \lambda_i + \mathbf{Q}_i
\end{equation}
in which $L(\mathbf{q}_i,\dot{\mathbf{q}}_i) = T(\mathbf{q}_i,\dot{\mathbf{q}}_i) - U(\mathbf{q}_i)$ is the Lagrangian of the $i^\text{th}$ vehicle, $U(\mathbf{q}_i)$ is the potential energy of the vehicle agent, $\lambda \in \mathbb{R}^{k \times 1}$ is the Lagrangian multiplier, and $A^T \lambda$ is   the constraint force.
$\mathbf{Q}_i = B(\mathbf{q}_i) [\tau_i - \mathbf{f}(\mathbf{u}_i)]$ denotes the external force, where $\tau_i$ is the force generated by the actuator, and $\mathbf{f}(\mathbf{u}_i)$ is the friction on the actuator.
Then equation~(\ref{Lagrangian_eq}) can be rewritten into
\begin{equation} \label{dynamics_raw}
	M(\mathbf{q}_i) \ddot{\mathbf{q}}_i + \dot{M} \dot{\mathbf{q}}_i - \left( \frac{\partial T_i}{\partial \mathbf{q}_i} \right)^T + \left( \frac{\partial U_i}{\partial \mathbf{q}_i} \right)^T + B(\mathbf{q}_i) \mathbf{f}(\dot{\mathbf{q}}_i) = A(\mathbf{q}_i) \lambda_i + B(\mathbf{q}_i) \tau_i
\end{equation}

By setting $C(\mathbf{q}_i, \dot{\mathbf{q}}_i) \dot{\mathbf{q}}_i = \dot{M} \dot{\mathbf{q}}_i - \left( \frac{\partial T_i}{\partial \mathbf{q}_i} \right)^T$, $F (\mathbf{q}_i, \dot{\mathbf{q}}_i) = B(\mathbf{q}_i) \mathbf{f}(\dot{\mathbf{q}}_i)$, and $G(\mathbf{q}_i) = \left( \frac{\partial U_i}{\partial \mathbf{q}_i} \right)^T$, equation~(\ref{dynamics_raw}) can be thereby transferred into (\ref{dynamics_q}).
Notice that the form of $C_{n \times n}$ is not unique, however, with a proper definition of the matrix $C$, we will have $\dot{M} - 2 C$ to be skew-symmetric.
The $(i,j)^\text{th}$ entry of $C$ is defined as follows \cite{book_robotics}
\begin{equation}
	c_{ij} = \sum_{k = 1}^n c_{ijk} \dot{q}_k
\end{equation}
where $\dot{q}_k$ is the $k^\text{th}$ entry of $\dot{\mathbf{q}}$, and $c_{ijk} = \frac{1}{2} \left( \frac{\partial m_{ij}}{\partial q_k} + \frac{\partial m_{ik}}{\partial q_j} - \frac{\partial m_{jk}}{\partial q_i} \right)$ is defined using the Christoffel symbols of the first kind.
Then we have the centripetal and Coriolis matrix calculated as
$
	C = \mtx{0 & 0 & -m d \dot{\theta}_i \cos \theta_i \\
			0 & 0 & -m d \dot{\theta}_i \sin \theta_i \\
			0 & 0 & 0 \\}$.
Since the vehicle is operating on the ground, the gravity vector $G$ is equal to zero.
The friction vector $F$ is assumed to be a nonlinear function of the general velocity $\mathbf{u}_i$, and is unknown to the controller.

To eliminate the nonholonomic constraint force $A(\mathbf{q}_i) \lambda_i$ from equation~(\ref{dynamics_q}), we left multiplying $J^T(\mathbf{q}_i)$ to the equation, it yields:
\begin{equation}
	J^T M J \dot{\mathbf{u}}_i + J^T (M \dot J + C J) \mathbf{u}_i + J^T F + J^T G = J^T B \tau_i + J^T A \lambda_i
\end{equation}
From equation~(\ref{constraint}) and (\ref{kinematics}), we have $J^T A = \mathbf{0}_{2 \times 1}$, then the dynamic equation of $\mathbf{u}_i$ is simplified as
\begin{equation} \label{dynamics_v}
	\bar{M} (\mathbf{q}_i) \dot{\mathbf{u}}_i + \bar{C} (\mathbf{u}_i) \mathbf{u}_i + \bar{F} (\mathbf{u}_i) + \bar{G} (\mathbf{q}_i) = \bar{\tau}_i,
\end{equation}
where
\begin{equation*}
	\begin{split}
		&\bar M = J^T M J = \mtx{m & 0 \\ 0 & m d^2 + I \\}, \quad
		\bar C = J^T (M \dot J + C J) = \mtx{0 & -m d \dot{\theta}_i \\ m d \dot{\theta}_i & 0}, \\
		&\bar F = J^T F, \quad
		\bar G = J^T G = \mathbf{0}_{2 \times 1}, \quad
		\bar \tau_i = \mtx{\bar \tau_{vi} \\ \bar \tau_{\omega i}}
		= J^T B \tau_i = \mtxr{1/r & 1/r \\ R/r & -R/r} \tau_i. \\
	\end{split}
\end{equation*}
The degree of freedom of the vehicle dynamics is now reduced to two.
Since $J^T B$ is of full rank, then for any transformed torque input $\bar \tau_i$, there exists a unique corresponding actual torque input $\tau_i \in \mathbb R^2$ that applied on each wheel.

The main challenge for controlling the system includes i) the direct measurement of the linear and angular velocities is not feasible, and ii) system parameter matrices $\bar{C}$ and $\bar{F}$ are unknown to the controller.

Based on the above system setup, we are ready to formulate our objective of this chapter.
Consider a group of $n$ homogeneous unicycle-type vehicles, the kinematics and dynamics of each vehicle agent are described by equations~(\ref{kinematics}) and (\ref{dynamics_v}), respectively.
The communication graph of such $n$ vehicles is denoted as $\mathcal{G}$.
Regarding this MAS, we have the following assumption.
	
\begin{assumption} \label{assu_G}
The graph $\mathcal G$ is undirected and connected.
\end{assumption}

The objective of this chapter is to design an output-feedback adaptive learning control law for each vehicle agent in the MAS, such that
\begin{enumerate}[label=\roman*)]
	\item \textit{State estimation:}
	The immeasurable general velocities $\mathbf{u}_i = \mtx{v_i & \omega_i}^T$ can be estimated by a high-gain observer using the measurement of the general coordinates $\mathbf{q}_i = \mtx{x_i & y_i & \theta_i}^T$.
	\item \textit{Trajectory tracking:}
	Each vehicle in the MAS will track its desired reference trajectory, which will be quantified by $(x_{ri} (t), y_{ri} (t), \theta_{ri} (t))$; i.e., $\lim_{t \rightarrow \infty} (x_i (t) - x_{ri} (t)) = 0$, $\lim_{t \rightarrow \infty} (y_i (t) - y_{ri} (t)) = 0$, $\lim_{t \rightarrow \infty} (\theta_i (t) - \theta_{ri} (t)) = 0$.
	\item \textit{Cooperative Learning:}
	The unknown homogeneous dynamics of all the vehicles can be locally accurately identified along the union of the trajectories experienced by all vehicle agents in the MAS.
	\item \textit{Experience based control:}
	The identified/learned knowledge from the cooperative learning phase can be re-utilized by each local vehicle to perform stable trajectory tracking with improved control performance.
\end{enumerate}

In order to apply the deterministic learning theory, we have the following assumption on the reference trajectories.
\begin{assumption} \label{assu_PE}
The reference trajectories $x_{ri} (t)$, $y_{ri} (t)$, $\theta_{ri} (t)$ for all $i = 1, \cdots, n$ are recurrent.
\end{assumption}

\section{Main results}

\subsection{High-gain observer design}

In mobile robotics control, the position of the vehicle can be easily obtained in real time using GPS signals or camera positioning, while the direct measurement of the velocities is much more difficult.
For the control and system estimation purposes, the velocities of the vehicle are required for the controller.
To this end, we follow the high-gain observer design method in \cite{khalil2008high, Main}, and introduce a high-gain observer to estimate the velocities using robot positions.
First, we define two new variables as follows
\begin{equation} \label{p_rotatry_frame}
	\begin{split}
		p_{x_i} &= x_i \cos \theta_i + y_i \sin \theta_i \\
		p_{y_i} &= y_i \cos \theta_i - x_i \sin \theta_i \\
	\end{split}
\end{equation}
Notice that the operation above can be considered as a projecting the vehicle position onto the a frame whose origin is fixed to the origin of ground coordinates, and the axes are parallel to the body-fixed frame of the vehicle.
The coordinates of the vehicle in this rotational frame is $(p_{x_i}, p_{y_i})$ and hence, $p_{x_i}$ and $p_{y_i}$ can be calculated based on the measurement of the position and the orientation.
The rotation rate of this frame equals to the angular velocity of the vehicle $\dot{\theta}_i = \omega_i$.
Based on this, we design the high-gain observer for $\omega$ as
\begin{equation} \label{observer_omega}
	\begin{split}
		\dot{\hat{\theta}}_i &= \hat{\omega}_i + \frac{l_1}{\delta} (\theta_i - \hat{\theta}_i) \\
		\dot{\hat{\omega}}_i &= \frac{l_2}{\delta^2} (\theta_i - \hat{\theta}_i) \\
	\end{split}
\end{equation}
in which $\delta$ is a small positive scalar to be designed, and $l_1$ and $l_2$ are parameters to be chosen, such that $\mtx{-l_1 & 1 \\ -l_2 & 0 \\}$ is Hurwitz stable.
The time derivative of this coordinates defined in (\ref{p_rotatry_frame}) is given by $\dot{p}_{x_i} = v_i + p_{y_i} \omega_i$, and $\dot{p}_{y_i} = -p_{x_i} \omega_i$,
then we design the high-gain observer for $v$ as
\begin{equation} \label{observer_v}
	\begin{split}
		\dot{\hat{p}}_{x_i} &= \hat{v}_i + p_{y_i} \hat{\omega}_i + \frac{l_1}{\delta} (p_{x_i} - \hat{p}_{x_i}) \\
		\dot{\hat{v}}_i &= \frac{l_2}{\delta^2} (p_{x_i} - \hat{p}_{x_i}) \\
	\end{split}
\end{equation}
To prevent peaking while using this high-gain observer and in turn improving the transient response, parameter $\delta$ cannot be too small \cite{Main}.
Due to the use of a globally bounded control, decreasing $\delta$ does not induce peaking phenomenon of the state variables of the system, while the ability to decrease $\delta$ will be limited by practical factors such as measurement noise and sampling rates \cite{lee1997adaptive, oh1997nonlinear}.
According to \cite{khalil2008high}, it is easy to show that the estimation error between the actual and estimated velocities of the $i^\text{th}$ vehicle $\mathbf{z}_i = \mathbf{u}_i - \hat{\mathbf{u}}_i$ will converge to zero, detailed proof is omitted here due to space limitation.

\subsection{Controller design and tracking convergence analysis} \label{Sec_tracking}

\begin{figure} [h]
	\centering
		\includegraphics[width=0.35\textwidth]{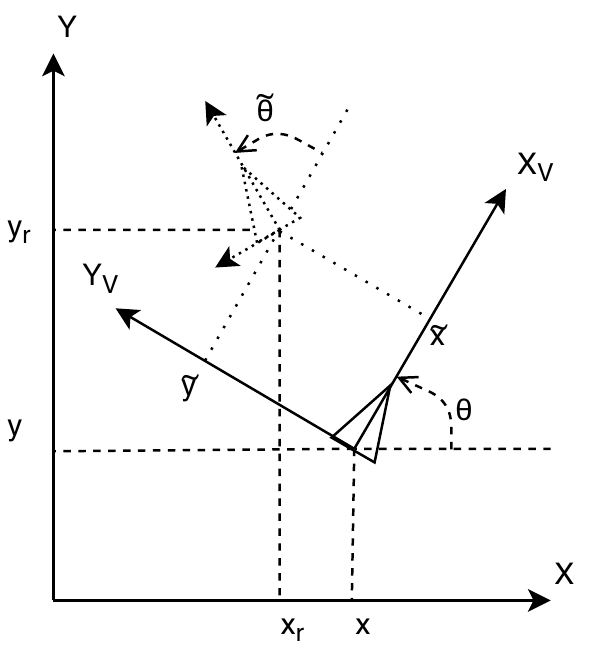}
	\caption{Projecting tracking error onto the body-fixed frame}
	\label{trajectory_tracking}
\end{figure}

After obtaining the linear and angular velocities from the high-gain observer, we now proceed to the trajectory tracking.
First, we define the tracking error $\tilde{\mathbf{q}}_i$ by projecting $\mathbf{q}_{ri} - \mathbf{q}_i$ onto the body coordinate of the $i^{th}$ vehicle, with the $x$ axis set to be the front and $y$ to be the left of the vehicle, as shown in Fig.~\ref{trajectory_tracking}.
\begin{equation}
	\tilde{\mathbf q}_i = \mtx{\tilde x_i \\ \tilde y_i \\ \tilde \theta_i}
	= \begin{bmatrix*}[r]
		\cos \theta_i & \sin \theta_i & 0 \\
		-\sin \theta_i & \cos \theta_i & 0 \\
		0 \quad & 0 \quad & 1 \\
		\end{bmatrix*}
		\mtx{x_{ri} - x_i \\ y_{ri} - y_i \\ \theta_{ri} - \theta_i},
\end{equation}
using the constraint~(\ref{constraint}) and kinematics~(\ref{kinematics}), we have the derivative of the tracking error as follows
\begin{equation} \label{error_dynamics_q}
	\begin{split}
		\dot{\tilde x}_i &= v_{ri} \cos \tilde \theta_i + \omega_i \tilde y_i - v_i \\
		\dot{\tilde y}_i &= v_{ri} \sin \tilde \theta_i - \omega_i \tilde x_i \\
		\dot{\tilde \theta}_i &= \omega_{ri} - \omega_{i} \\
	\end{split}
\end{equation}
where $v_i$ and $\omega_i$ are the linear and angular velocities of the $i^\text{th}$ vehicle, respectively.

In order to utilize the backstepping control theory, we treat $v_i$ and $\omega_i$ in equation~(\ref{error_dynamics_q}) as virtual inputs, then following the methodology from \cite{TargetVelocity}, we can design a stabilizing virtual controller as
\begin{equation} \label{uc}
	\mathbf{u}_{c_i} = \mtx{v_{c_i} \\ \omega_{c_i}}
	= \mtx{v_{r_i} \cos \tilde \theta_i + K_x \tilde x_i \\
		\omega_{r_i} + v_{r_i} K_y \tilde y_i + K_\theta \sin \tilde \theta_i \\},
\end{equation}
in which $K_x$, $K_y$, and $K_\theta$ are all positive constants.
It can be shown that this virtual velocity controller is able to stabilize the closed-loop system~(\ref{error_dynamics_q}) kinematically by replacing $v_i$ and $\omega_i$ with $v_{c_i}$ and $\omega_{c_i}$, respectively.
To this end, we define the following Lyapunov function for the $i^\text{th}$ vehicle
\begin{equation} \label{V1}
	V_{1_i} = \frac{\tilde x_i^2}{2} + \frac{\tilde y_i^2}{2} + \frac{(1 - \cos \tilde \theta_i)}{K_y}
\end{equation}
and the derivative of $V_{1_i}$ is
\begin{equation} \label{V1_dot}
	\begin{split}
		\dot V_{1_i}
		&= \tilde x_i \dot{\tilde x}_i + \tilde y_i \dot{\tilde y}_i
			+ \frac{\sin \tilde \theta_i}{K_y} \dot{\tilde \theta}_i \\
		&= \tilde x_i (v_{r_i} \cos \tilde \theta_i + \omega_i \tilde y_i - v_{c_i})
			+ \tilde y_i (v_{r_i} \sin \tilde \theta_i - \omega_i \tilde x_i)
			+ \frac{\sin \tilde \theta_i}{K_y} (\omega_{r_i} - \omega_{c_i}) \\
		&= \tilde x_i (\omega_i \tilde y_i - K_x \tilde x_i)
			+ \tilde y_i (v_{r_i} \sin \tilde \theta_i - \omega_i \tilde x_i)
			+ \frac{\sin \tilde \theta_i}{K_y} (-v_{r_i} K_y \tilde y_i - K_\theta \sin \tilde \theta_i) \\
		&= -K_x \tilde x_i^2 - \frac{K_\theta}{K_y} \sin^2 \tilde \theta_i \leq 0\\
	\end{split}
\end{equation}

Since $\dot{V}_{1i}$ is negative semi-definite, then we can conclude that this closed-loop system is stable, i.e., the tracking error $\tilde{\mathbf{q}}_i$ for the $i^\text{th}$ vehicle will be bounded.
\begin{remark}
In addition to the stable conclusion above, we could also conclude the asymptotic stability by finding the invariant set of $\dot{V}_{1_i} = 0$.
By setting $\dot{V}_{1_i} = 0$, we have $\tilde{x}_i = 0$ and $\sin \tilde{\theta} = 0$.
Applying this result into equation~(\ref{error_dynamics_q}) and (\ref{uc}), we have the invariant set equals to $\{ \tilde{x}_i = 0, \tilde{y}_i = 0, \sin \tilde{\theta} = 0 \} \cup \{ \tilde{x}_i = 0, \sin \tilde{\theta} = 0, \tilde{y}_i = 0, v_{r_i} = 0, \omega_{ri} = 0 \}$.
With the assumption~\ref{assu_PE}, the velocity of the reference cannot be constant over time, then we can conclude that the only invariant subset of $\dot{V}_{1_i} = 0$ is the origin $\tilde{\mathbf{q}}_i = \mathbf{0}$.
Therefore, we can conclude that the closed-loop system (\ref{error_dynamics_q}) and (\ref{uc}) is asymptotically stable \cite{book_robust}.
\end{remark}

With the idea of backstepping control, we then derive the transformed torque input $\bar{\tau}_i$ for the $i^{th}$ vehicle with the following steps.
By defining the error between the virtual controller $\mathbf{u}_{c_i}$ and the actual velocity $\mathbf{u}_i$ as $\tilde{\mathbf u}_i = \mtx{\tilde{v}_i & \tilde{\omega}_i}^T = \mathbf u_{c_i} - \mathbf u_i$, we can rewrite equation~(\ref{error_dynamics_q}) in terms of $\tilde{v}_i$ and $\tilde{\omega}_i$ as
\begin{equation} \label{error_dynamics_q_u}
	\begin{split}
		\dot{\tilde x}_i &= v_{r_i} \cos \tilde \theta_i + \omega_i \tilde y_i - v_{c_i} + \tilde{v}_i
			= -K_x \tilde x_i + \omega_i \tilde y_i + \tilde{v}_i \\
		\dot{\tilde y}_i &= -\omega_i \tilde x_i + v_{r_i} \sin \tilde \theta_i \\
		\dot{\tilde \theta}_i &= \omega_{r_i} - \omega_{c_i} + \tilde{\omega}_i
			= -v_{r_i} K_y \tilde y_i - K_\theta \sin \tilde \theta_i + \tilde{\omega}_i \\
	\end{split}
\end{equation}

Then we define a new Lyapunov function $V_{2_i} = V_{1_i} + \frac{\tilde{\mathbf{u}}_i^T \bar{M} \tilde{\mathbf{u}}_i}{2}$ for the closed-loop system~(\ref{error_dynamics_q_u}), whose derivative can be written as
\begin{equation} \label{V2_dot}
	\begin{split}
		\dot V_{2_i}
		&= \tilde x_i \dot{\tilde x}_i + \tilde y_i \dot{\tilde y}_i
			+ \frac{\sin \tilde \theta_i}{K_y} \dot{\tilde \theta}_i
			+ \tilde{\mathbf{u}}_i^T \bar{M} \dot{\tilde{\mathbf{u}}}_i \\
		&= \tilde x_i (-K_x \tilde x_i + \omega_i \tilde y_i + \tilde{v}_i)
			+ \tilde y_i (-\omega_i \tilde x_i + v_{r_i} \sin \tilde \theta_i)
			+ \frac{\sin \tilde \theta_i}{K_y} (-v_{r_i} K_y \tilde y_i - K_\theta \sin \tilde \theta_i + \tilde{\omega}_i) \\
		&\quad + \tilde{\mathbf{u}}_i^T \bar{M} \dot{\tilde{\mathbf{u}}}_i \\
		&= -K_x \tilde x_i^2 - \frac{K_\theta}{K_y} \sin^2 \tilde \theta_i
			+ \tilde{\mathbf{u}}_i^T \left( \mtx{\tilde{x}_i \\ \frac{\sin \tilde \theta_i}{K_y}} + \bar{M} \dot{\tilde{\mathbf{u}}}_i \right) \\
	\end{split}
\end{equation}


To make the system stable, the term $\tilde{\mathbf{u}}_i^T \left( \mtx{\tilde{x}_i \\ \frac{\sin \tilde \theta_i}{K_y}} + \bar{M} \dot{\tilde{\mathbf{u}}}_i \right)$ needs to be negative definite.
From the definition of $\tilde{\mathbf{u}}_i$ and equation~(\ref{dynamics_v}), we have
\begin{equation} \label{M_u_dot}
	\begin{split}
		\bar{M} \dot{\tilde{\mathbf{u}}}_i = \bar{M} \dot{\tilde{\mathbf{u}}}_{c_i} - \bar{M} \dot{\mathbf{u}}_i
		= \bar{M}\dot{\tilde{\mathbf{u}}}_{c_i} + \bar C \mathbf{u}_i + \bar F - \bar{\tau}_i \\
	\end{split}
\end{equation}

Motivated from the results of \cite{Main}, it is easy to show that this term is negative definite if $\bar{\tau}_i$ is designed to be
\begin{equation} \label{input_ideal}
	\bar \tau_i = \bar M \dot{\mathbf{u}}_{c_i} + \bar C \mathbf{u}_i + \bar F + K_u \tilde{\mathbf{u}}_i + \mtx{\tilde x_i \\ \frac{\sin \tilde \theta_i}{K_y}\\},
\end{equation}
where $K_u$ is a positive constant.
Since the actual linear and angular velocity of the vehicle is unknown, we use $\hat{v}_i$ and $\hat{\omega}_i$ generated by the high-gain observer~(\ref{observer_omega}) and (\ref{observer_v}) to replace $v_i$ and $\omega_i$ in equation~(\ref{input_ideal}).
From the discussion in previous subsection, the convergence of velocities estimation is guaranteed.

In equation~(\ref{input_ideal}), 
$\bar{C} (\mathbf{u}_i)$ and $\bar{F} (\mathbf{u}_i)$ are unknown to the controller.
To overcome this issue, RBFNN will be used to approximate this nonlinear uncertain term, i.e.,
\begin{equation} \label{unknown_dynamics}
	H(X_i) = \bar{C} (\mathbf{u}_i) \mathbf{u}_i + \bar{F} (\mathbf{u}_i) = W^{*T} S(X_i) + \epsilon_i,
\end{equation}
in which $S(X_i)$ is the vector of RBF, with the variable (RBFNN input) $X_i = \mathbf{u}_i$, $W^*$ is the common ideal estimation weight of this RBFNN, and $\epsilon_i$ is the ideal estimation error, which can be made arbitrarily small given sufficiently large number of neurons.
Consequently, we proposed the implementable controller for the $i^{th}$ vehicle as follows
\begin{equation} \label{input_actual}
	\bar{\tau}_i = \bar{M} \dot{\mathbf{u}}_{c_i} + \hat{W}_i^T S(X_i) + K_u \mtx{v_{c_i} - \hat{v}_i \\ \omega_{c_i} - \hat{\omega}_i} + \mtx{\tilde x_i \\ \frac{\sin \tilde \theta_i}{K_y}},
\end{equation}

For the NN weights used in equation~(\ref{input_actual}), we propose an online NN weight updating law as follows
\begin{equation} \label{weight_updating_law}
	\dot {\hat W}_i = \Gamma S(X_i) \tilde{\mathbf u}_i^T - \gamma \hat W_i - \beta \sum_{j=1}^{n} a_{ij} (\hat W_i - \hat W_j),
\end{equation}
where $\Gamma$, $\gamma$, and $\beta$ are positive constants.

\begin{theorem} \label{Thm_tracking}
Consider the closed-loop system consisting of the $n$ vehicles in the MAS described by equation~(\ref{kinematics}) and (\ref{dynamics_v}), reference trajectory $\mathbf{q}_{r_i} (t)$, high-gain observer~(\ref{observer_omega}) and (\ref{observer_v}), adaptive NN controller (\ref{input_actual}) with the virtual velocity (\ref{uc}), and the online weight updating law (\ref{weight_updating_law}), under the assumptions~\ref{assu_G} and \ref{assu_PE}, then for any bounded initial condition of all the vehicles and $\hat{W}_i = 0$, the tracking error $\tilde{\mathbf{q}}_i$ converges asymptotically to a small neighborhood around zero for all vehicle agents in the MAS.
\end{theorem}

\begin{proof}
We first derive the error dynamics of velocity between $\mathbf u_{c_i}$ and $\mathbf u_i$
using equation~(\ref{M_u_dot}) and (\ref{input_actual})
\begin{equation} \label{error_dynamics_u}
	\dot{\tilde{\mathbf{u}}}_i
	= \bar{M}^{-1} \left[ \tilde W_i^T S(X_i) + \mathbf{\epsilon}_i
		- K_u \mtx{v_{c_i} - \hat{v}_i \\ \omega_{c_i} - \hat{\omega}_i}
		- \mtx{\tilde x_i \\ \frac{\sin \tilde \theta_i}{K_y}} \right]
\end{equation}
where $\epsilon_i = \mtx{\epsilon_{v_i} & \epsilon_{\omega_i}}^T$ and $\tilde W_i = W^* - \hat{W}_i$.
Notice that the convergence of $\hat{\mathbf{u}}_i$ to $\mathbf{u}_i$ is guaranteed by the high-gain observer.
Then we derive the error dynamics of NN weight as follows
\begin{equation} \label{error_dynamics_W}
	\begin{split}
		\dot{\tilde W}_i &= -\dot{\hat W}_i = -\Gamma S(X_i) \tilde{\mathbf u}_i^T + \gamma \hat W_i
			+ \beta \sum_{j=1}^{n} a_{ij} (\hat W_i - \hat W_j)) \\
	\end{split}
\end{equation}

For the closed-loop system given by equation~(\ref{error_dynamics_q_u}), (\ref{error_dynamics_u}), and~(\ref{error_dynamics_W}), we can build a positive definite function $V$ as
\begin{equation} \label{V}
	V = \sum_{i=1}^{n} \left[ \frac{\tilde x_i^2}{2} + \frac{\tilde y_i^2}{2}
			+ \frac{(1 - \cos \tilde \theta_i)}{K_y}
			+ \frac{\tilde{\mathbf{u}}_i^T \bar{M} \tilde{\mathbf{u}}_i}{2}
			+ \frac{\operatorname{trace}(\tilde W_i^T \tilde W_i)}{2 \Gamma} \right]
\end{equation}
whose derivative is equal to
\begin{equation}
	\dot V = \sum_{i=1}^{n} \left[ \tilde x_i \dot{\tilde x}_i + \tilde y_i \dot{\tilde y}_i
		+ \frac{\sin \tilde \theta_i}{K_y} \dot{\tilde \theta}_i
		+ \tilde{\mathbf{u}}_i^T \bar{M} \dot{\tilde{\mathbf{u}}}_i
		+ \frac{\operatorname{trace}(\tilde W_i^T \dot{\tilde{W}}_i)}{\Gamma} \right]
\end{equation}

By using equations~(\ref{error_dynamics_u}) and (\ref{error_dynamics_W}), the equation above is equivalent to
\begin{equation} \label{V_dot}
	\begin{split}
		\dot V
		&= \sum_{i=1}^{n} \left\{ \tilde x_i (\tilde v_i + \omega_i \tilde y_i - K_x \tilde x_i)
			+ \tilde y_i (v_{r_i} \sin \tilde \theta_i - \omega_i \tilde x_i)
			+ \frac{\sin \tilde \theta_i}{K_y} (\tilde \omega_i - v_{r_i} K_y \tilde y_i - K_\theta \sin \tilde \theta_i) \right. \\
		&\qquad + \tilde{\mathbf{u}}_i^T
			\left[ \tilde W_i^T S(X_i) + \mathbf{\epsilon}_i - K_u \tilde{\mathbf{u}}_i
			- \mtx{\tilde x_i \\ \frac{\sin \tilde \theta_i}{K_y}} \right] \\
		&\qquad \left. + \operatorname{trace} \left( \tilde W_i^T
			\left[ - S(X_i) \tilde{\mathbf u}_i^T + \frac{\gamma \hat W_i}{\Gamma}
			+ \frac{\beta}{\Gamma} \sum_{j=1}^{n} a_{ij} (\hat W_i - \hat W_j)) \right]
			\right) \right\} \\
		&= \sum_{i=1}^{n} \left\{ -K_x \tilde{x}_i^2 - \frac{K_\theta}{K_y} \sin^2 \tilde{\theta}_i
			- K_u \tilde{\mathbf{u}}_i^T \tilde{\mathbf{u}}_i
			+ \tilde{\mathbf{u}}_i^T \mathbf{\epsilon}_i + \tilde{\mathbf{u}}_i^T [\tilde W_i^T S(X_i)] \right. \\
		&\qquad \left. - \operatorname{trace} \left( [\tilde W_i^T S(X_i)] \tilde{\mathbf u}_i^T \right)
			+ \operatorname{trace} \left( \frac{\gamma \tilde W_i^T \hat W_i}{\Gamma} \right) \right\}
			-\operatorname{trace} \left( \sum_{i=1}^{n}
			\frac{\beta}{\Gamma} \tilde W_i^T \sum_{j=1}^{n} a_{ij} (\hat W_i - \hat W_j)) \right) \\
		&= \sum_{i=1}^{n} \left\{ -K_x \tilde{x}_i^2 - \frac{K_\theta}{K_y} \sin^2 \tilde{\theta}_i
			- K_u \tilde{\mathbf{u}}_i^T \tilde{\mathbf{u}}_i
			+ \tilde{\mathbf{u}}_i^T \mathbf{\epsilon}_i
			+ \frac{\gamma}{\Gamma} \operatorname{trace} \left( \tilde W_i^T \hat W_i \right) \right\}
			- \frac{\beta}{\Gamma} \operatorname{trace} \left( \tilde W^T (L \otimes I) \tilde W \right) \\
	\end{split}
\end{equation}
where $L$ is the Laplacian matrix of $\mathcal{G}$, and $\tilde{W} = \mtx{\tilde{W}_1^T & \cdots & \tilde{W}_n^T}^T$.
Since $\beta$ and $\Gamma$ are all positive, and $L$ is positive semi-definite, then we have $\frac{\beta}{\Gamma} \operatorname{trace} \left( \tilde W^T (L \otimes I) \tilde W \right) \geq 0$.
Notice that the estimation error can be made arbitrary small with a sufficient large number of neurons, and $\gamma$ is the leakage term chosen as a small positive constant.
Therefore, we can conclude that the closed-loop system (\ref{error_dynamics_q_u}), (\ref{error_dynamics_u}), and (\ref{error_dynamics_W}) is stable, i.e. $\dot{V} \leq 0$, if the following condition stands
\begin{equation}
	K_x \tilde{x}_i^2 + \frac{K_\theta}{K_y} \sin^2 \tilde{\theta}_i + K_u \tilde{\mathbf{u}}_i^T \tilde{\mathbf{u}}_i
	\geq \tilde{\mathbf{u}}_i^T \mathbf{\epsilon}_i + \frac{\gamma}{\Gamma} \operatorname{trace} \left( \tilde W_i^T \hat W_i \right)
\end{equation}
Hence, the closed-loop system is stable, and all tracking error are bounded.
Since all variables in (\ref{V_dot}) are continuous (i.e. $\ddot{V}$ is bounded), then with the application of Barbalat's lemma \cite{barbalat1959systemes}, we have $\lim_{t \rightarrow \infty} \dot{V} = 0$, which implies that the tracking error $\tilde{\mathbf{q}}_i$ for all agents will converge to a small neighborhood of zero, whose size depends on the norm of $\tilde{\mathbf{u}}_i^T \mathbf{\epsilon}_i + \frac{\gamma}{\Gamma} \operatorname{trace} \left( \tilde W_i^T \hat W_i \right)$.
\end{proof}

\subsection{Consensus convergence of NN weights} \label{Sec_learning}

In addition to the tracking convergence shown in the previous subsection, we will show that all vehicles in the system is able to learn the unknown vehicle dynamics along the union trajectory (denoted as $\cup_{i=1}^n \zeta_i [X_i(t)]$) experienced by all vehicles in this subsection.

By defining $\tilde v = \mtx{\tilde v_1 & \dots & \tilde v_n}^T$, $\tilde \omega = \mtx{\tilde \omega_1 & \dots & \tilde \omega_n}^T$, $\tilde W_v = \mtx{\tilde W_{1,1} & \dots & \tilde W_{n,1}}^T$, and $\tilde W_\omega = \mtx{\tilde W_{1,2} & \dots & \tilde W_{n,2}}^T$, we combine the error dynamics in equations~(\ref{error_dynamics_u}) and (\ref{error_dynamics_W}) for all vehicles into the following form:
\begin{equation} \label{error_dynamics}
	\mtx{\dot{\tilde v} \\ \dot{\tilde \omega} \\ \dot{\tilde W}_v \\ \dot{\tilde W}_\omega}
	= \mtx{A & B \\ C & D \\} \mtx{\tilde v \\ \tilde \omega \\ \tilde W_v \\ \tilde W_\omega} + E
\end{equation}
in which
\begin{equation*}
	\begin{array}{rclrcl}
			A_{2n \times 2n}
		&=& \mtx{-\frac{K_u}{m} I_n & 0 \\ 0 & -\frac{K_u}{I} I_n \\},
		&	B_{2nN \times 2n}
		&=& \mtx{\frac{\mathbf{S}^T}{m} & 0 \\ 0 & \frac{\mathbf{S}^T}{I} \\}, \\
			C_{2n \times 2nN}
		&=& \mtx{-\Gamma \mathbf{S} & 0 \\ 0 & -\Gamma \mathbf{S} \\},
		&	D_{2nN \times 2nN}
		&=& \mtx{-\beta (L \otimes I_N) & 0 \\ 0 & -\beta (L \otimes I_N) \\}, \\
	\end{array}
\end{equation*}
where $\mathbf{S} = \operatorname{diag}(S(X_1), S(X_2), \dots, S(X_n))$, and
\begin{equation*}
\begin{aligned}
	E_{(2nN + 2nN) \times 1} &= \mtx{E_1 \\ E_2 \\ E_3 \\ E_4}, \quad
	E_1 = \frac{1}{m} \mtx{\epsilon_{v_1} - \tilde x_1 \\
			\vdots \\ \epsilon_{v_n} - \tilde x_n \\}, \quad
	E_2 = \frac{1}{I} \mtx{\epsilon_{\omega_1} - \frac{\sin \tilde \theta_1}{K_y} \\
			\vdots \\ \epsilon_{\omega_n} - \frac{\sin \tilde \theta_n}{K_y}}, \\
	E_3 &= \frac{\gamma}{m} \mtx{\hat{W}_{1,1} \\ \vdots \\ \hat{W}_{n,1} \\}, \quad
	E_4 = \frac{\gamma}{m} \mtx{\hat{W}_{1,2} \\ \vdots \\ \hat{W}_{n,2} \\}.
\end{aligned}
\end{equation*}

As is shown in Theorem~\ref{Thm_tracking}, the tracking error $\tilde{\mathbf{q}}_i$ will converge to a small neighborhood of zero for all vehicle agents in the MAS.
Furthermore, the ideal estimation errors $\epsilon_{vi}$ and $\epsilon_{\omega i}$ can be made arbitrarily small given sufficient number of RBF neurons, and $\gamma$ is chosen to ba a small positive constant, therefore, we can conclude that the norm of $E$ in equation~(\ref{error_dynamics}) is a small value.
In the following theorem, we will show that $W_i = \mtx{W_{i,1} & W_{i,2}}$ converges to a small neighborhood of the common ideal weight $W^*$ for all $i = 1, \dots, n$ under assumptions \ref{assu_G} and \ref{assu_PE}.

Before proceeding further, we denote the system trajectory of the $i^{th}$ vehicle as $\zeta_{i}$ for all $i = 1, \cdots, n$.
Using the same notation from \cite{DL}, $(\cdot)_{\zeta}$ and $(\cdot)_{\bar{\zeta}}$ represent the parts of $(\cdot)$ related to the region close to and away from the trajectory $\zeta$, respectively.

\begin{theorem} \label{Thm_learning}
Consider the error dynamics (\ref{error_dynamics}), under the assumptions~\ref{assu_G} and \ref{assu_PE}, then for any bounded initial condition of all the vehicles and $\hat{W}_i = 0$, along the union of the system trajectories $\cup_{i=1}^n \zeta_i [X_i(t)]$, all local estimated neural weights $\hat W_{\zeta_i}$ used in (\ref{input_actual}) and (\ref{weight_updating_law}) converge to a small neighborhood of their common ideal value $W_{\zeta}^*$, and locally accurate identification of nonlinear uncertain dynamics $H(X(t))$ can be obtained by $\hat{W}_i^T S(X)$ as well as $\bar{W}_i^T S(X)$ for all $X \in \cup_{i=1}^n \zeta_i [X_i(t)]$, where
\bgeq \label{W_bar}
	\bar W_i = \operatorname*{mean}_{t_{a_i} \leq t \leq t_{b_i}} \hat W_i (t)
\edeq
with $[t_{a_i}, t_{b_i}]$ ($t_{b_i} > t_{a_i} > T_i$)  being a time segment after the transient period of tracking control.
\end{theorem}

\begin{proof}
According to \cite{DL}, if the nominal part of closed loop system shown in (\ref{error_dynamics}) is uniformly locally exponentially stable (ULES), then $\tilde v$, $\tilde \omega$, $\tilde W_v$, and $\tilde W_\omega$ will converge to a small neighborhood of the origin, whose size depends on the value of $||E||$.

Now the problem boils down to proving ULES of the nominal part of system~(\ref{error_dynamics}).
To this end, we need to resort to the results of Lemma 4 in \cite{Cooperative_PE}.
It is stated that if the Assumptions 1 and 2 therein are satisfied, and the associated vector $S_{\zeta} (X_i)$ is PE for all $i = 1, \cdots, n$, then the nominal part of (\ref{error_dynamics}) is ULES.
The assumption 1 therein is automatically verified since $\mathbf{S}$ is bounded, and Assumption 2 therein also holds, if we set the counterparts $P = \Gamma \mtx{m & 0 \\ 0 & I \\}$ and $Q = -2 \Gamma \mtx{K_v I_n & 0 \\ 0 & K_\omega I_n \\}$.
Furthermore, the PE condition of $S_\zeta (X_i)$ will also be met, if $X_i$ of the learning task is recurrent \cite{DL}, which is guaranteed by Assumption \ref{assu_PE} and results from Theorem \ref{Thm_tracking}.
Therefore, we can obtain the conclusion that $\tilde v$, $\tilde \omega$, $\tilde W_v$, and $\tilde W_\omega$ will converge to a small neighborhood of the origin, whose size depends on the small value of $||E||$.

Similar to \cite{wang2006learning}, the convergence of $\hat{W}_{\zeta i}$ to a small neighborhood of $W_\zeta^*$ implies that for all $X \in \cup_{i=1}^n \zeta_i [X_i(t)]$, we have
\begin{equation} \label{estimation_*2hat}
	\begin{split}
		H(X) &= W_\zeta^{*T} + \epsilon_\zeta
			= \hat{W}_{\zeta_i}^T S_\zeta (X) + \tilde{W}_{\zeta_i}^T S_\zeta (X) + \epsilon_{\zeta i}
			= \hat{W}_{\zeta_i}^T S_\zeta (X) + \epsilon_{1 \zeta_i} \\
	\end{split}
\end{equation}
where $\epsilon_{1 \zeta i} = \tilde{W}_{\zeta i}^T S_\zeta (X) + \epsilon_{\zeta i}$ is close to $\epsilon_{\zeta i}$ due to the convergence of $\tilde{W}_{\zeta i}$.
With the $\bar{W}_i$ defined in (\ref{W_bar}), then equation~(\ref{estimation_*2hat}) can be rewritten into
\begin{equation} \label{estimation_hat2bar}
	\begin{split}
		H(X) &= \hat{W}_{\zeta_i}^T S_\zeta (X) + \epsilon_{1 \zeta_i}
			= \bar{W}_{\zeta_i}^T S_\zeta (X) + \epsilon_{2 \zeta_i} \\
	\end{split}
\end{equation}
where $\bar{W}_{\zeta_i}^T = \mtx{w_{1_\zeta} & \cdots & w_{k_\zeta}}^T$ is a subvector of $\bar{W}_i$ and $\epsilon_{2 \zeta_i}$ is the error using $\bar{W}_{\zeta i}^T S_\zeta (X)$ as the system approximation.
After the transient process, $||\epsilon_{1 \zeta_i}|| - ||\epsilon_{2 \zeta_i}||$ is small for all $i = 1, \cdots, n$.

On the other hand, due to the localization  property of Gaussian RBFs, both $S_{\bar{\zeta}}$ and $\bar{W}_{\bar{\zeta}} S_{\bar{\zeta}} (X)$ are very small.
Hence, along the union trajectory $\cup_{i=1}^n \zeta_i [X_i(t)]$, the entire constant RBF network $\bar{W}^T S(X)$ can be used to approximate the nonlinear uncertain dynamics, demonstrated by the following equivalent equations
\begin{equation} \label{estimation_actual}
	\begin{split}
		H(X) &= W_\zeta^{*T} S_\zeta (X) + \epsilon_\zeta \\
		H(X) &= \hat{W}_{\zeta_i}^T S_\zeta (X) + \hat{W}_{\bar{\zeta}_i}^T S_{\bar{\zeta}} (X) + \epsilon_{1_i} = \hat{W}_i^T S(X) + \epsilon_{1_i} \\
		H(X) &= \bar{W}_{\zeta_i}^T S_\zeta (X) + \bar{W}_{\bar{\zeta}_i}^T S_{\bar{\zeta}} (X) + \epsilon_{2_i} = \bar{W}_i^T S(X) + \epsilon_{2_i} \\
	\end{split}
\end{equation}
where $||\epsilon_{1_i}|| - ||\epsilon_{1 \zeta_i}||$ and $||\epsilon_{2_i}|| - ||\epsilon_{2 \zeta_i}||$ are all small for all $i = 1, \cdots, n$.
Therefore, the conclusion of Theorem \ref{Thm_learning} can be drawn.
\end{proof}

\subsection{Experience-based trajectory tracking control} \label{Sec_experience}

In this section, based on the learning results from the previous subsections, we further propose an experience-based trajectory tracking control method using the knowledge learned in the previous subsection, such that the experience-based controller is able to drive each vehicle to follow any reference trajectory experienced by any vehicle on the learning stage.

To this end, we replace the NN weight $\hat{W}_i$ in equation~(\ref{input_actual}) by the converged constant NN weight $\bar{W}_i$ for the $i^{th}$ vehicle.
Therefore, the experience-based controller for the $i^{th}$ vehicle is constructed as follows
\begin{equation} \label{input_exp}
	\bar{\tau}_i = \bar{M} \dot{\mathbf{u}}_{c_i} + \bar{W}_i^T S(X_i) + K_u \mtx{v_{c_i} - \hat{v}_i \\ \omega_{c_i} - \hat{\omega}_i} + \mtx{\tilde x_i \\ \frac{\sin \tilde \theta_i}{K_y}},
\end{equation}
in which $\dot{\mathbf{u}}_{ci}$ is the derivative of the virtual velocity controller from equation~(\ref{uc}), and $\bar W_i$ is obtained from equation~(\ref{W_bar}) for the $i^{th}$ vehicle.
%
The system model (\ref{kinematics}) and (\ref{dynamics_v}), and the high-gain observer design~(\ref{observer_v}) and (\ref{observer_omega}) remain unchanged.

\begin{theorem} \label{Thm_experience}
Consider the closed-loop system consisting of equation~(\ref{kinematics}) and (\ref{dynamics_v}), reference trajectory $\mathbf{q}_{ri} \in \cup_{j=1}^n \mathbf{q}_j(t)$, high-gain observer~(\ref{observer_v}) and (\ref{observer_omega}), and the experience-based controller (\ref{input_exp}) with virtual velocity (\ref{uc}).
For any bounded initial condition, 
the tracking error $\tilde{\mathbf{q}}_i$ converges asymptotically to a small neighborhood around zero.
\end{theorem}

\begin{proof}
Similar to the proof of Theorem~\ref{Thm_tracking}, by defining $\tilde{\mathbf{q}}_i$ and $\tilde{\mathbf{u}}_i$ to be the error between the position and velocity of the $i^\text{th}$ vehicle and its associated reference trajectory, we have the error dynamics of the $i^{th}$ vehicle as
\begin{equation}
	\begin{split}
		\dot{\tilde{x}}_i &= v_{r_i} \cos \tilde{\theta}_i + \omega_i \tilde{y}_i - v_i
			= \tilde{v}_i + \omega_i \tilde{y}_i - K_x \tilde{x}_i\\
		\dot{\tilde{y}}_i &= v_{r_i} \sin \tilde{\theta}_i - \omega_i \tilde{x}_i \\
		\dot{\tilde{\theta}}_i &= \omega_{r_i} - \omega_{i}
			= \tilde{\omega}_i - v_{r_i} K_y \tilde{y}_i - K_\theta \sin \tilde{\theta}_i \\
		\dot{\tilde{\mathbf{u}}}_i &= \bar{M}^{-1} \left[ H(X_i) - \bar{W}_i^T S(X_i)
			- K_u \mtx{v_{c_i} - \hat{v}_i \\ \omega_{c_i} - \hat{\omega}_i}
			- \mtx{\tilde x_i \\ \frac{\sin \tilde \theta_i}{K_y}} \right] \\
	\end{split}
\end{equation}
With the same high-gain observer design used in the learning-based tracking, the convergence of $\hat{\mathbf{u}}_i$ to $\mathbf{u}_i$ is also guaranteed.
For the closed-loop system shown above, we can build a positive definite function as
\begin{equation} \label{Lyapunov_exp}
	V_i = \frac{\tilde{x}_i^2}{2} + \frac{\tilde{y}_i^2}{2} + \frac{1 - \cos \tilde{\theta}_i}{K_y} + \frac{\tilde{\mathbf{u}}_i^T \bar{M} \tilde{\mathbf{u}}_i}{2}
\end{equation}
and the derivative of $V_i$ is
\begin{equation}
	\begin{split}
		\dot{V}_i
		&= \tilde{x}_i \dot{\tilde{x}}_i + \tilde{y}_i \dot{\tilde{y}}_i
			+ \frac{\sin \tilde{\theta}_i}{K_y} \dot{\tilde{\theta}}_i
			+ \tilde{\mathbf{u}}_i^T \bar{M} \dot{\tilde{\mathbf{u}}}_i \\
		&= \tilde{x}_i (\tilde{v}_i + \omega_i \tilde{y}_i - K_x \tilde{x}_i)
			+ \tilde{y}_i (v_{r_i} \sin \tilde{\theta}_i - \omega_i \tilde{x}_i)
			+ \frac{\sin \tilde{\theta}i}{K_y} (\tilde{\omega}_i - v_{r_i} K_y \tilde{y}_i - K_\theta \sin \tilde{\theta}_i) \\
		&\quad + \tilde{\mathbf{u}}_i^T \left( \mathbf{\epsilon}_{2i} - K_u \tilde{\mathbf{u}}_i - \mtx{\tilde{x}_i \\ \frac{\sin \tilde{\theta}_i}{K_y}} \right) \\
		&= -K_x \tilde{x}_i^2 - \frac{K_\theta}{K_y} \sin^2 \tilde{\theta}_i
			- K_u \tilde{\mathbf{u}}_i^T \tilde{\mathbf{u}}_i
			+ \tilde{\mathbf{u}}_i^T \mathbf{\epsilon}_{2i} \\
	\end{split}
\end{equation}
where $\mathbf{\epsilon}_{2i} = H(X_i) - \bar W_i^T S(X_i)$.
Then following the similar arguments in the proof of Theorem~\ref{Thm_tracking}, given positive $K_x$, $K_y$, $K_\theta$, and $K_u$, then we can conclude that the Lyapunov function $V_i$ is positive definite and $\dot V_i$ is negative semi-definite in the region $K_x \tilde{x}_i^2 + \frac{K_\theta}{K_y} \sin^2 \tilde{\theta}_i + K_u \tilde{\mathbf{u}}_i^T \tilde{\mathbf{u}}_i \geq \tilde{\mathbf{u}}_i^T \bar{\mathbf{\epsilon}}_i$.
Similar to the proof of Theorem~\ref{Thm_tracking}, it can be shown that $\lim_{t \rightarrow \infty} \dot{V_i} = 0$ with Barbalat's lemma, and the tracking errors will converge to a small neighborhood of zero.
\end{proof}


\section{Simulation Studies}


Consider four identical vehicles, whose unknown friction vector is assumed to be a nonlinear function of $v$ and $\omega$ as follows
$
	\bar F = \mtx{0.1 m v_i + 0.05 m v_i^2 \\ 0.2 I \omega_i + 0.1 I \omega_i^2}$,
and since we assume the vehicles are operating on the horizontal plane, the gravitational vector $\bar{G}$ is equal to zero.
The physical parameters of the vehicles are given as $m = 2 \, \text{kg}$, $I = 0.2 \, \text{kg} \cdot \text{m}^2$; $R = 0.15 \: \text{m}$, $r = 0.05 \: \text{m}$.
The the reference trajectories of the three vehicles are given by
\begin{equation*}
	\begin{array}{llll}
		\left\{
			\begin{aligned}
				x_{r_1} &= -\sin t \\
				y_{r_1} &= 2\cos t \\
			\end{aligned}
		\right.
		&\left\{
			\begin{aligned}
				x_{r_2} &= 2\cos t \\
				y_{r_2} &= \sin t \\
			\end{aligned}
		\right.
		&\left\{
			\begin{aligned}
				x_{r_3} &= -2\sin t \\
				y_{r_3} &= 3\cos t \\
			\end{aligned}
		\right.
		&\left\{
			\begin{aligned}
				x_{r_4} &= 3\cos t \\
				y_{r_4} &= 2\sin t \\
			\end{aligned}
		\right.
	\end{array}
\end{equation*}
and for all vehicles, the orientations of reference trajectories and vehicle velocities satisfy the following equations
\begin{equation*}
	\tan \theta_{ri} = \frac{\dot y_{ri}}{\dot x_{ri}},
	v_{ri} = \sqrt{\dot x_{ri}^2 + \dot y_{ri}^2},
	\omega_{ri} = \frac{\dot x_{ri} \ddot y_{ri} - \ddot x_{ri} \dot y_{ri}}{\dot x_{ri}^2 + \dot y_{ri}^2}.
\end{equation*}
The parameters of the observer~(\ref{observer_omega}) and (\ref{observer_v}) are given as $\epsilon = 0.01$, and $l_1 = l_2 = 1$.
The parameters of the controller~(\ref{input_actual}) with (\ref{uc}) are given as $K_x = K_y = K_\theta = 1$, and $K_u = 2$.
The parameters of (\ref{weight_updating_law}) are given as $\Gamma = 10$, $\gamma = 0.001$, and $\beta = 10$.
For each $i = 1, 2, 3, 4$, since $X_i = \mtx{v_i & \omega_i}^T$, we construct the Gaussian RBFNN $\hat{W}_i S(X_i)$ using $N = 5 \times 5 = 25$ neuron nodes with the centers evenly placed over the state space $[0, 4] \times [0, 4]$ and the standard deviation of the Gaussian function equal to $0.7$.
%
The initial position of the vehicles are set at the origin, with the velocities set to be zero,
and the initial weights of RBFNNs are also set to be zero.
The connection between three vehicles is shown in Figure~\ref{Connection}, {and the Laplacian matrix $L$ associated with the graph $\mathcal{G}$ is
\begin{equation*}
	L = \mtxr{2 & -1 &  0 & -1 \\
			-1 &  2 & -1 &  0 \\
			 0 & -1 &  2 & -1 \\
			-1 &  0 & -1 &  2 \\}.
\end{equation*}}
\begin{figure} [h]
	\centering
		\includegraphics[width=0.2\textwidth]{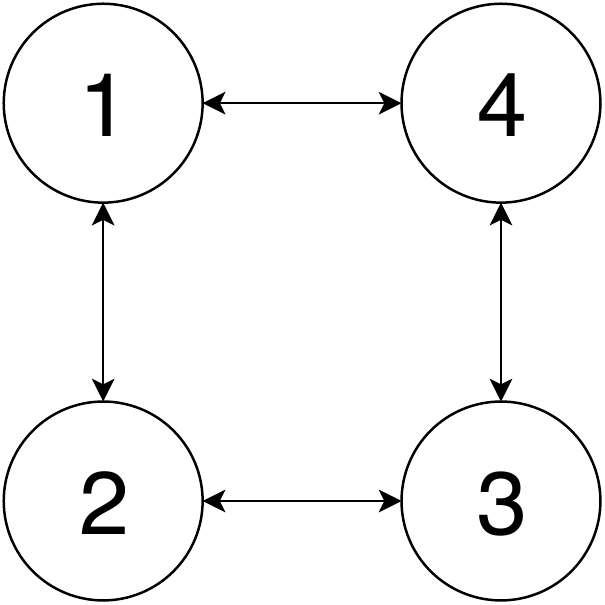}
	\caption{Connection between four vehicles}
	\label{Connection}
\end{figure}

\begin{figure*} [htb!]
	\centering
	\begin{subfigure}{0.45\textwidth}
		\includegraphics[width=\linewidth]{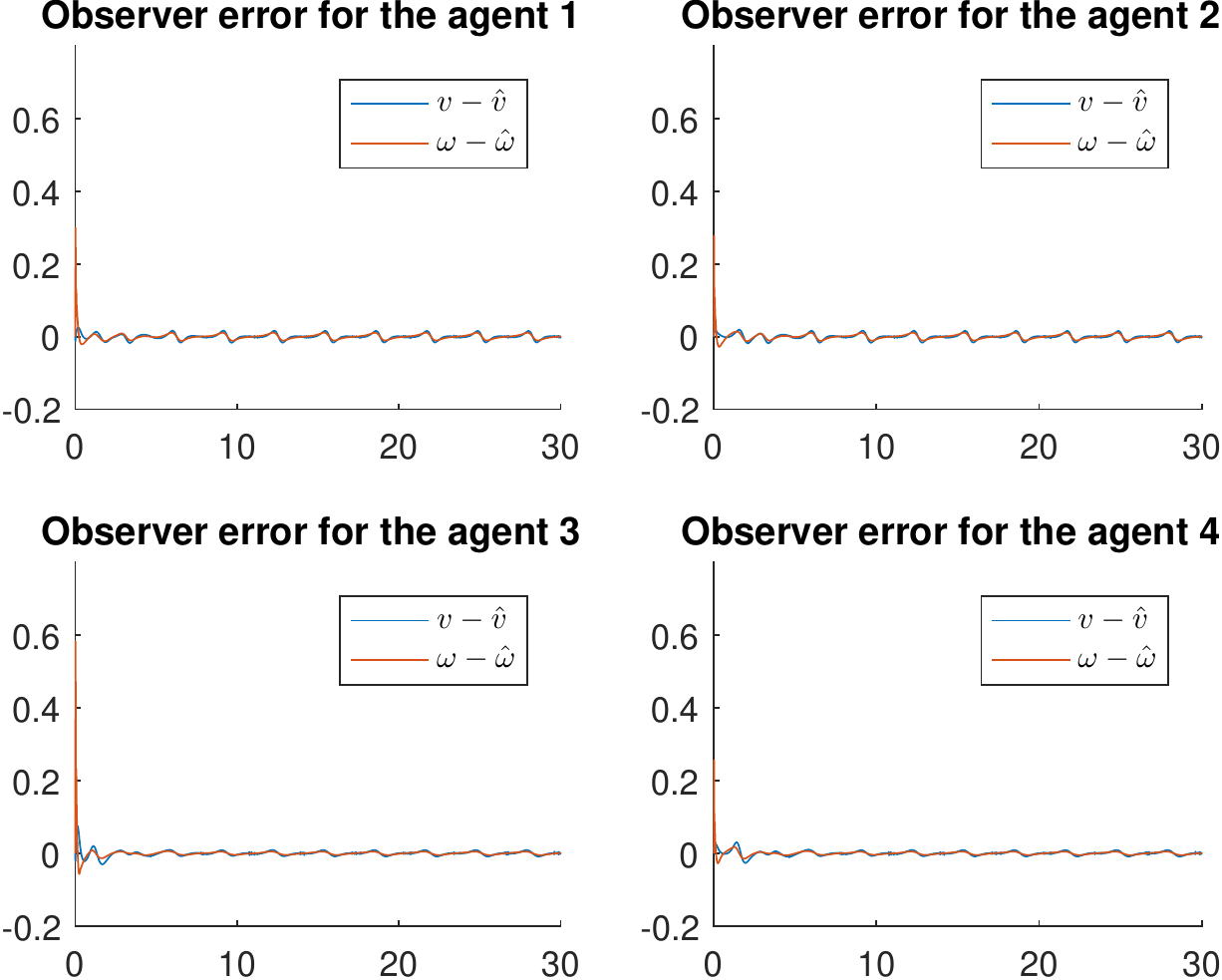}
		\caption{Observer errors using observer (\ref{observer_omega}) and (\ref{observer_v}).}
		\label{DLC_observer_error}
	\end{subfigure} \quad
	\begin{subfigure}{0.45\textwidth}
		\includegraphics[width=\linewidth]{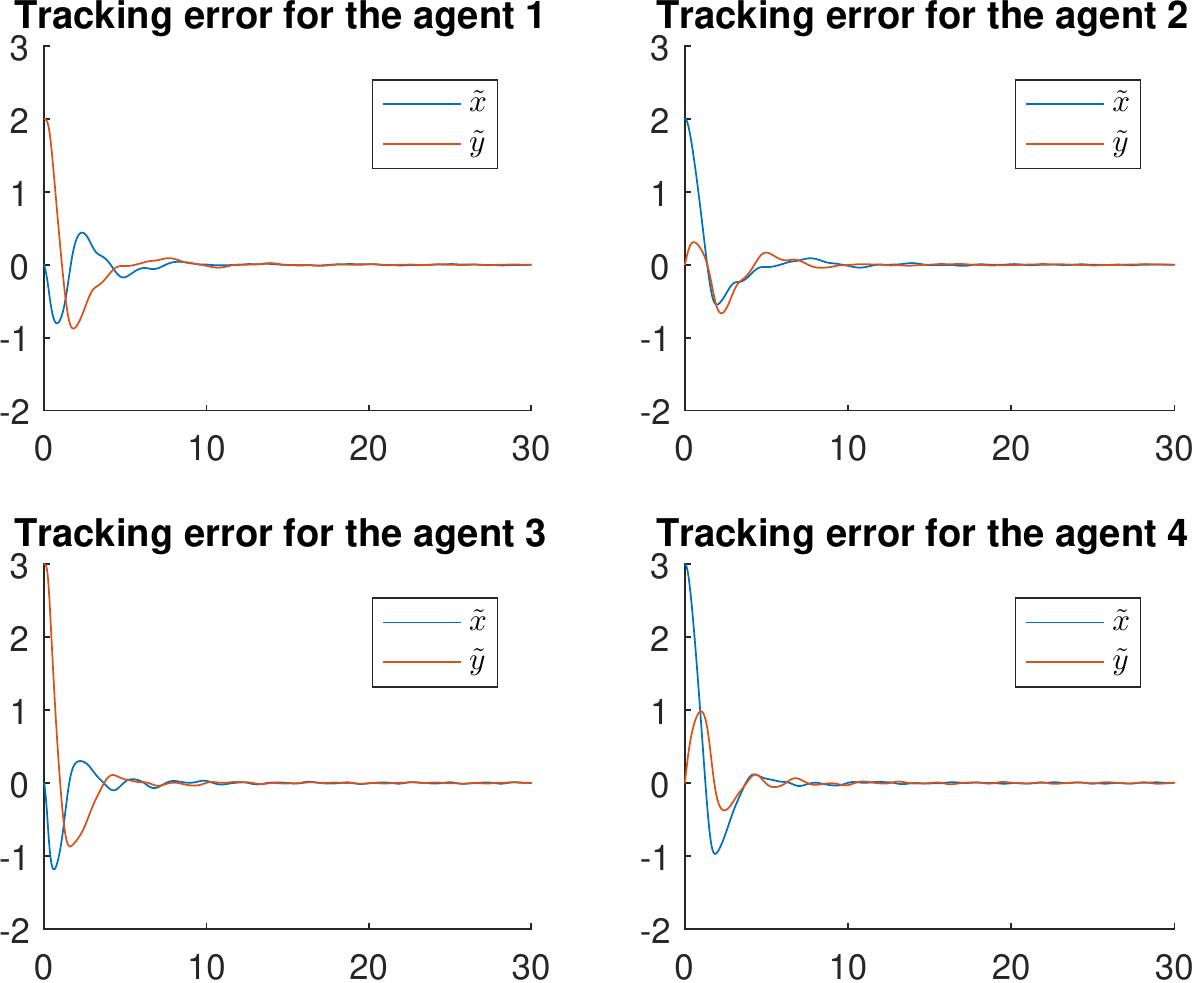}
		\caption{Tracking errors using controller (\ref{input_actual}) with (\ref{uc}) and (\ref{weight_updating_law}).}
		\label{DLC_tracking_error}
	\end{subfigure} \\
	\caption{Observer errors and tracking errors using observer-based controller.}
\end{figure*}

\begin{figure*} [htb!]
	\centering
	\begin{subfigure}{0.3\textwidth}
		\includegraphics[width=\linewidth]{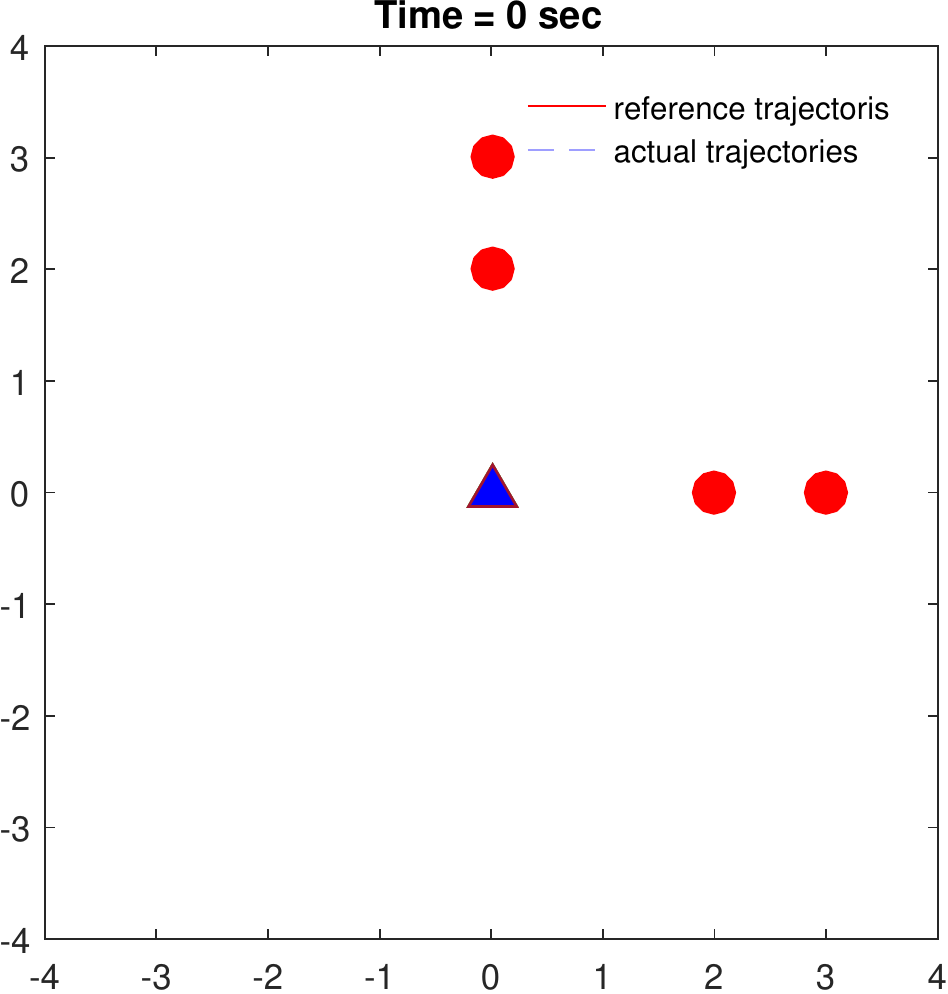}
		\caption{time at 0 seconds}
		\label{DLC_0s}
	\end{subfigure} \quad
	\begin{subfigure}{0.3\textwidth}
		\includegraphics[width=\linewidth]{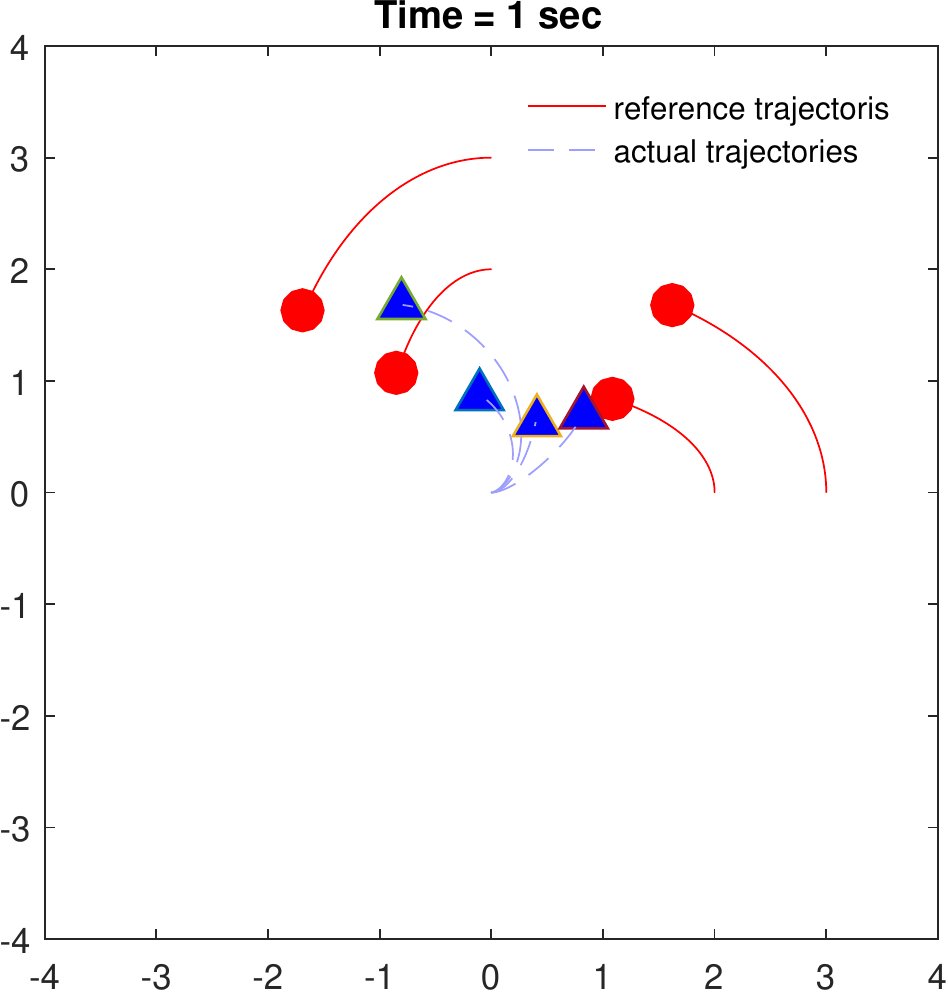}
		\caption{time at 1 seconds}
		\label{DLC_1s}
	\end{subfigure} \quad
	\begin{subfigure}{0.3\textwidth}
		\includegraphics[width=\linewidth]{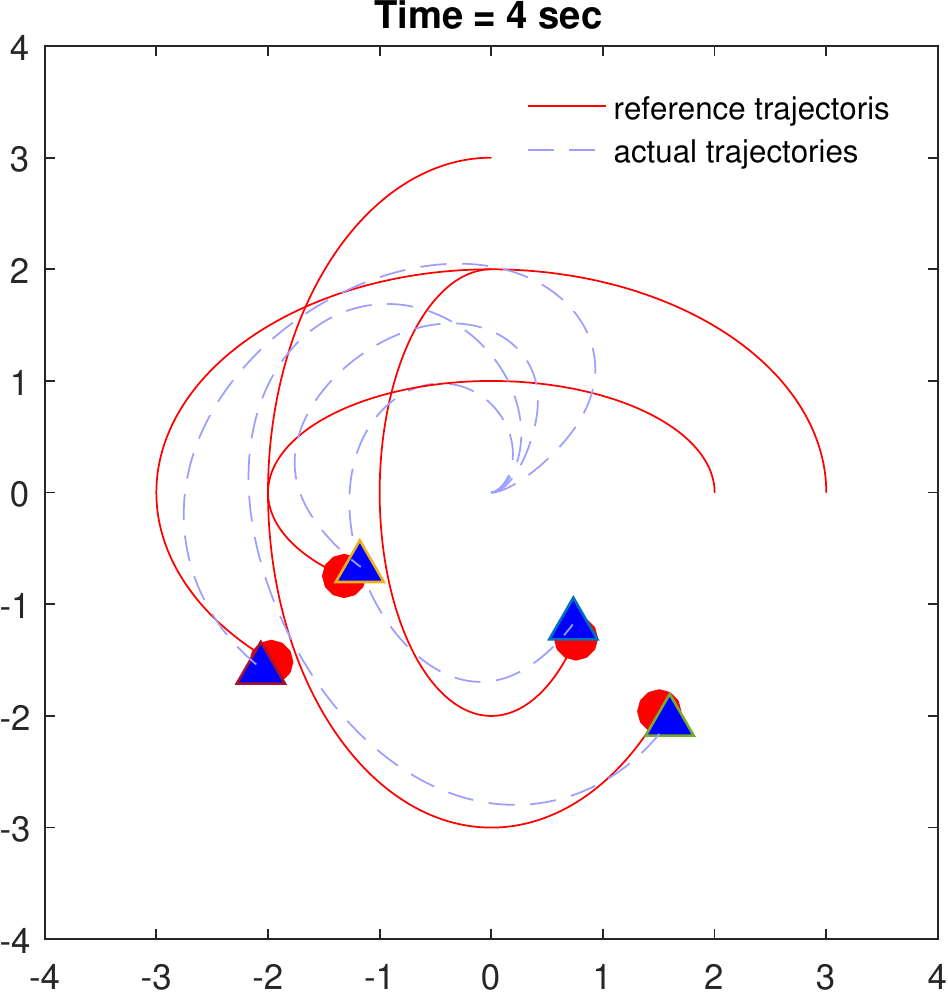}
		\caption{time at 4 seconds}
		\label{DLC_4s}
	\end{subfigure} \\
	\begin{subfigure}{0.3\textwidth}
		\includegraphics[width=\linewidth]{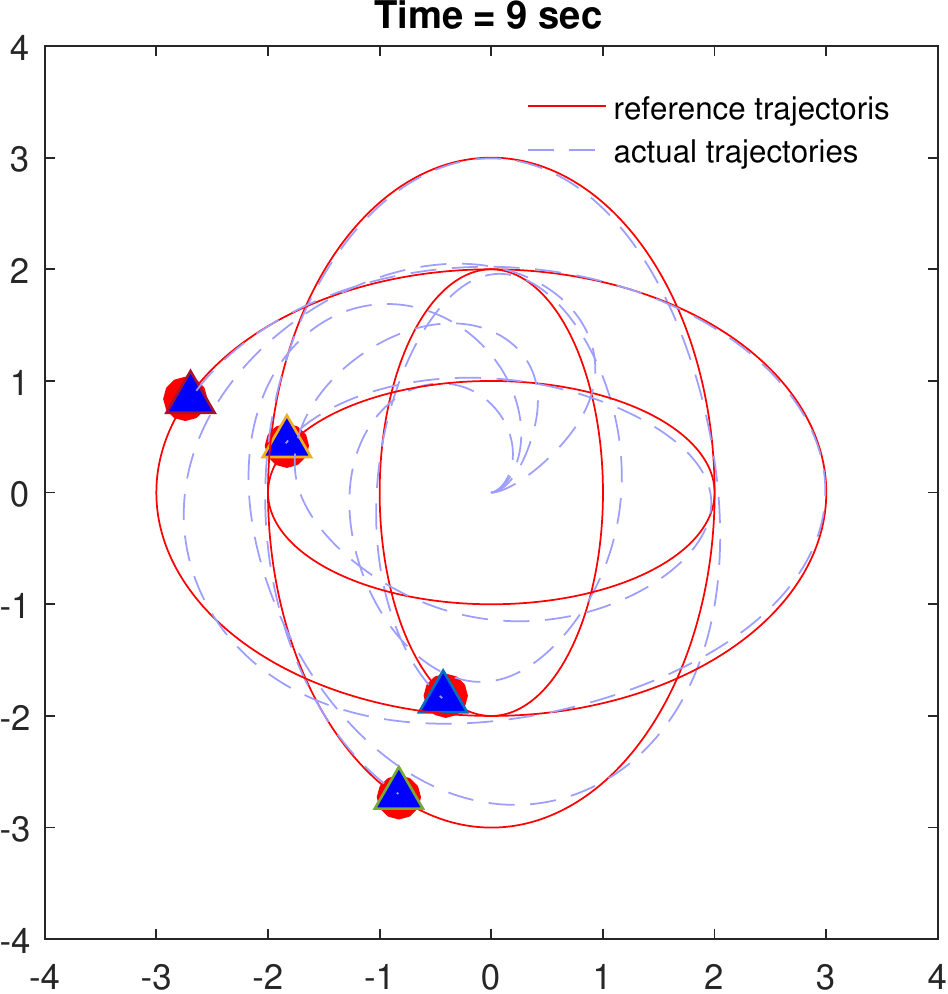}
		\caption{time at 9 seconds}
		\label{DLC_9s}
	\end{subfigure} \quad
	\begin{subfigure}{0.3\textwidth}
		\includegraphics[width=\linewidth]{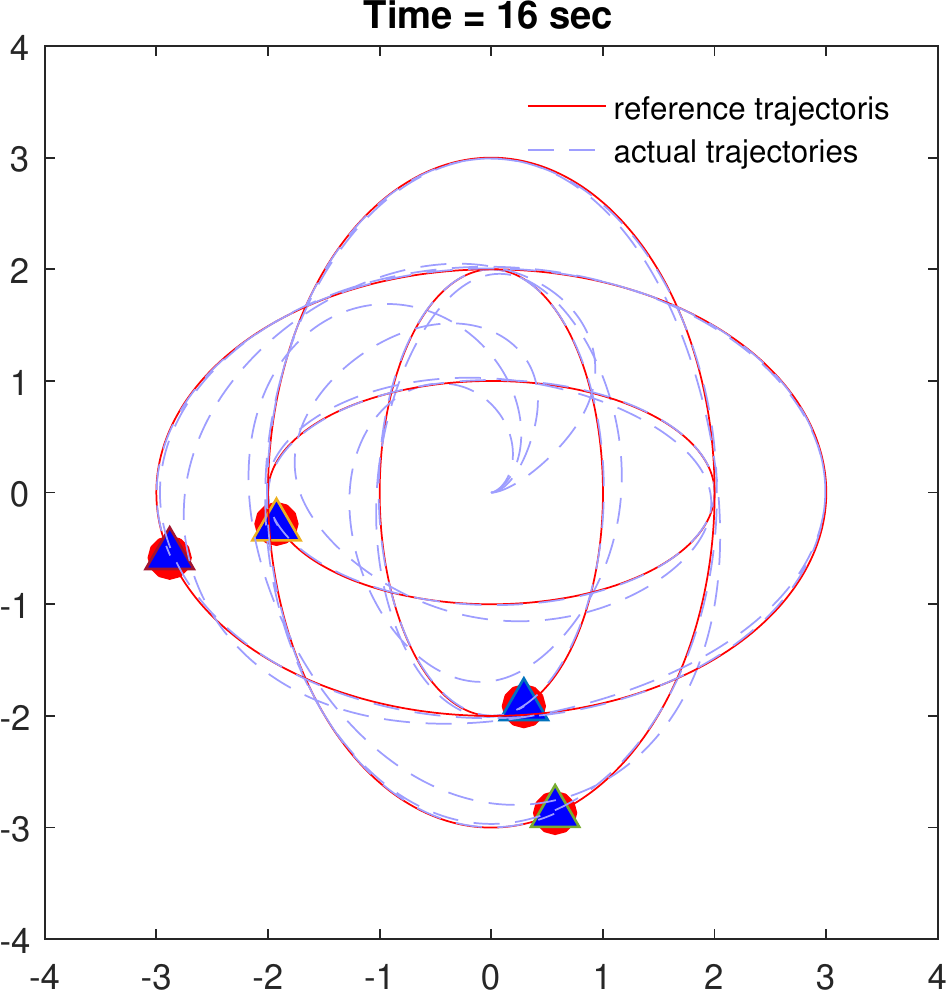}
		\caption{time at 16 seconds}
		\label{DLC_16s}
	\end{subfigure} \quad
	\begin{subfigure}{0.3\textwidth}
		\includegraphics[width=\linewidth]{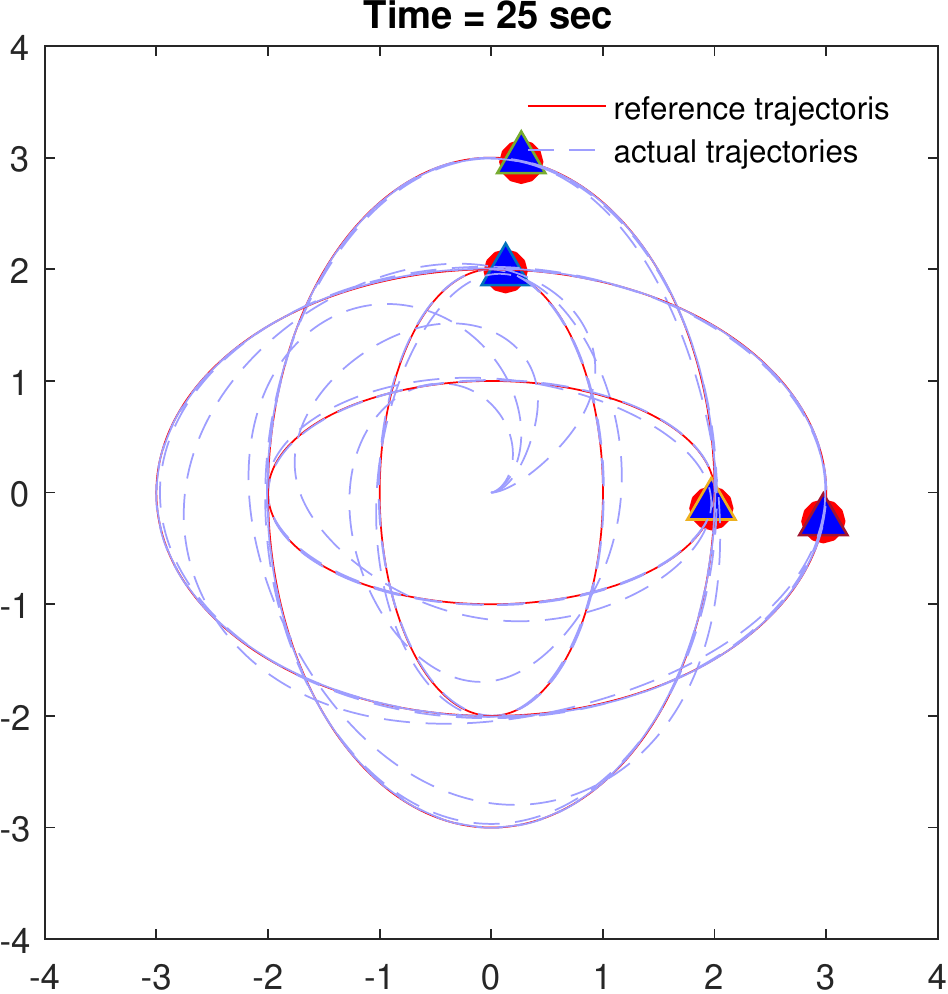}
		\caption{time at 25 seconds}
		\label{DLC_25s}
	\end{subfigure} \\
	\caption{Snapshot of trajectory tracking using controller (\ref{input_actual}) with (\ref{uc}) and (\ref{weight_updating_law}).}
	\label{2D_DLC}
\end{figure*}

\begin{figure} [htb!]
	\centering
	\begin{subfigure}{0.45\textwidth}
		\includegraphics[width=\linewidth]{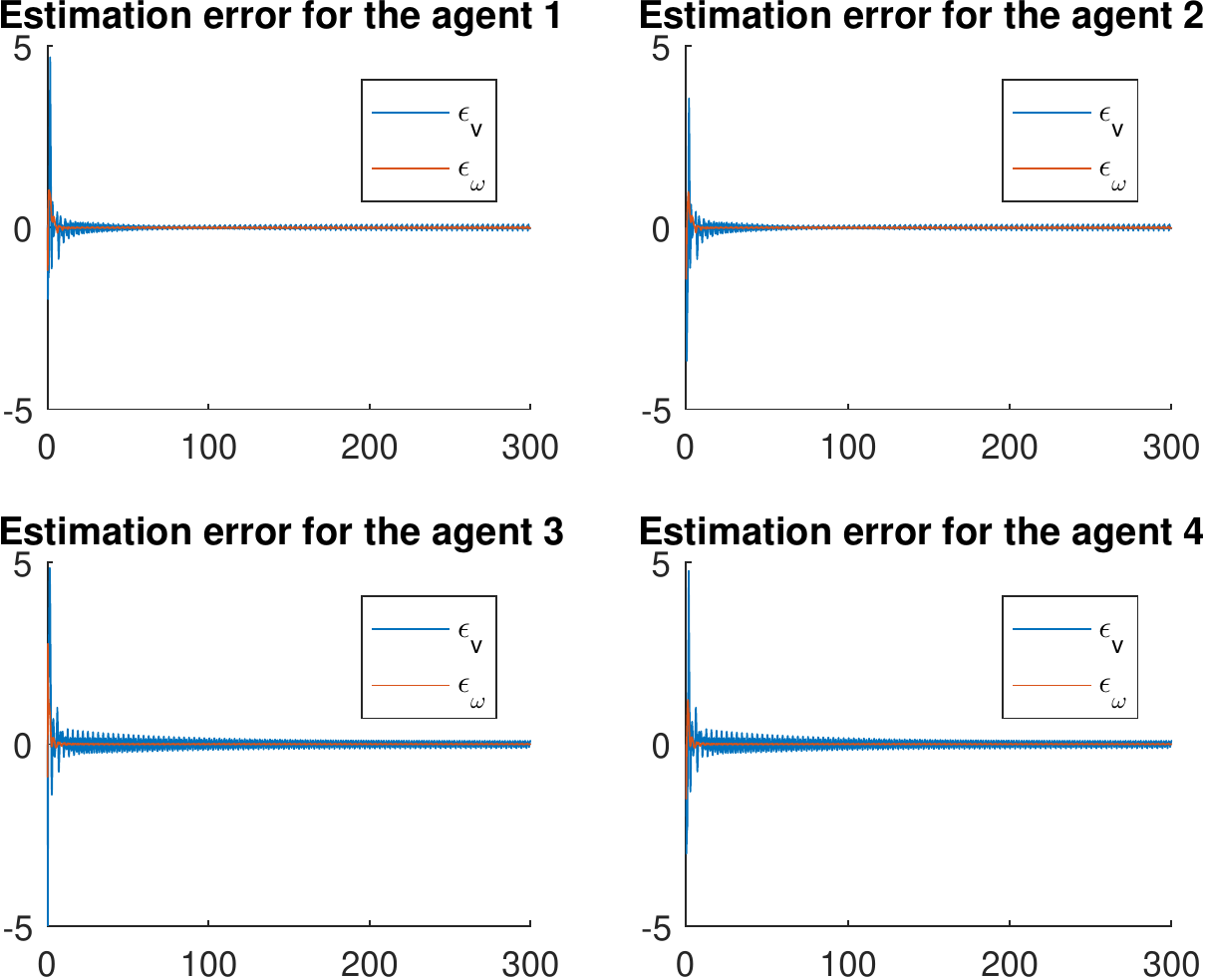}
		\caption{Estimation errors using controller (\ref{input_actual}) with (\ref{uc}) and (\ref{weight_updating_law}).}
		\label{DLC_estimation_error}
	\end{subfigure} \quad
	\begin{subfigure}{0.45\textwidth}
		\includegraphics[width=\linewidth]{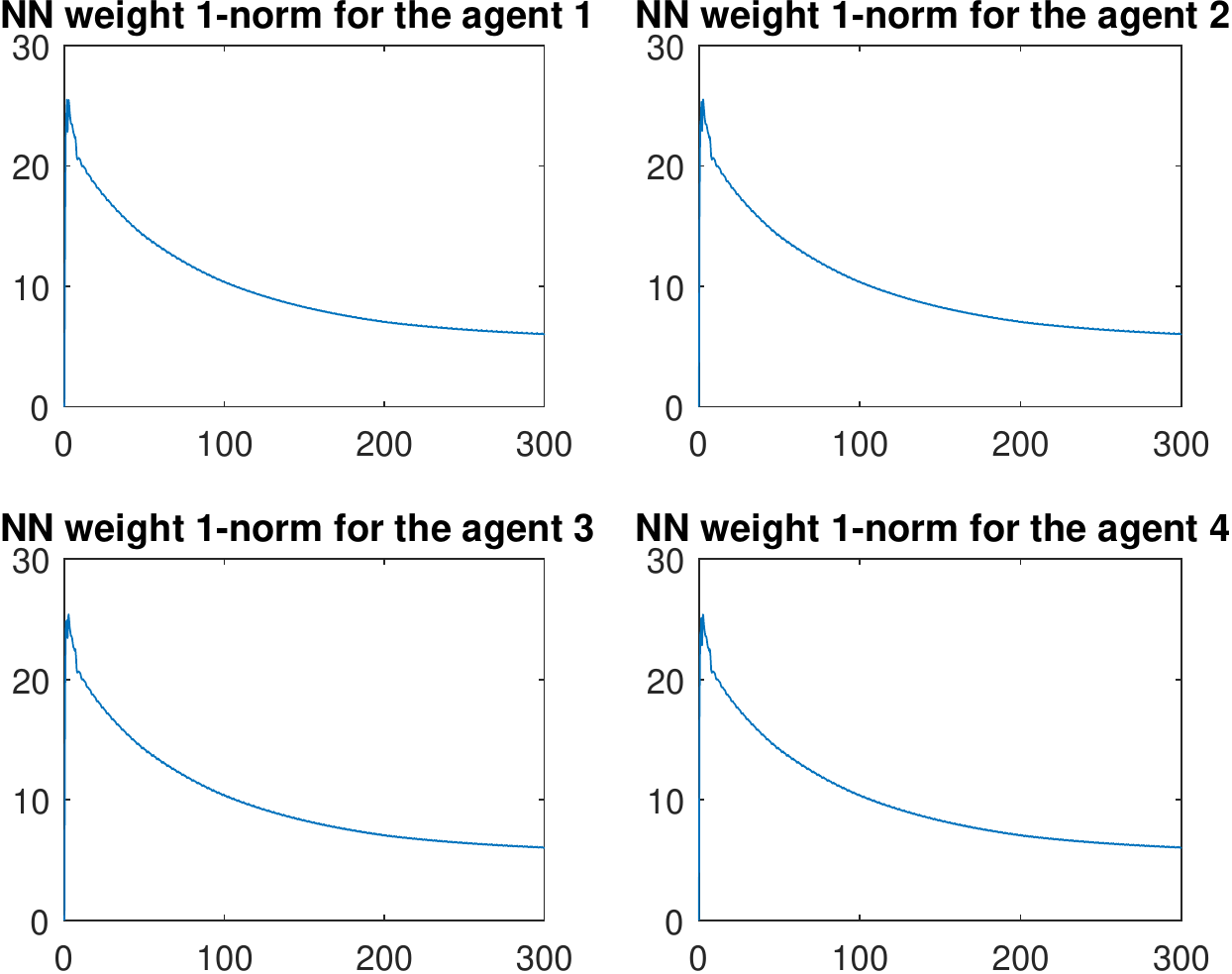}
		\caption{Weight vector 1-norm of $W_v$ and $W_\omega$}
		\label{DLC_weight_norm}
	\end{subfigure} \\
	\caption{Estimation errors and NN weight convergence.}
\end{figure}

Simulation results are shown as following.
Figure~\ref{DLC_observer_error} shows that the observer error will converge to a close neighborhood around zero in a very short time period, and
Figure~\ref{DLC_tracking_error} shows that all tracking errors $\tilde{x}_i$ and $\tilde{y}_i$ will converge to zero, and Figures~\ref{DLC_0s} to \ref{DLC_25s} show that all vehicles (blue triangles) will track its own reference trajectory (red solid circles) on the 2-D frame.
Figure~\ref{DLC_weight_norm} shows that the NN weights of all vehicle agents converge to the same constant, and Figure~\ref{DLC_estimation_error} shows that all RBFNNs of three vehicles are able to accurately estimate the unknown dynamics, as the estimation errors converging to a small neighborhood around zero.

\begin{figure*} [htb!]
	\centering
	\begin{subfigure}{0.45\textwidth}
		\includegraphics[width=\linewidth]{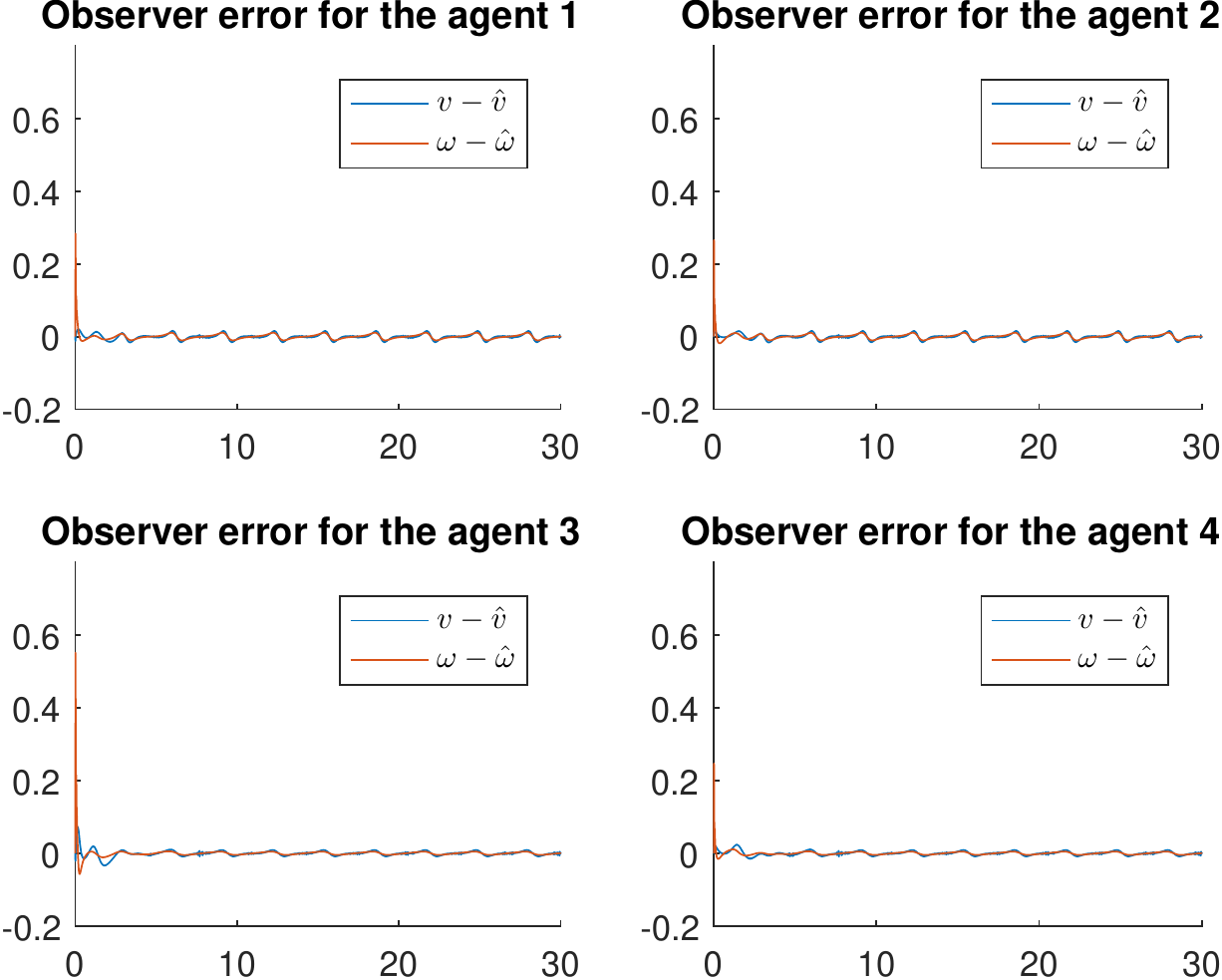}
		\caption{Observer errors using observer (\ref{observer_omega}) and (\ref{observer_v}).}
		\label{EXPC_observer_error}
	\end{subfigure} \quad
	\begin{subfigure}{0.45\textwidth}
		\includegraphics[width=\linewidth]{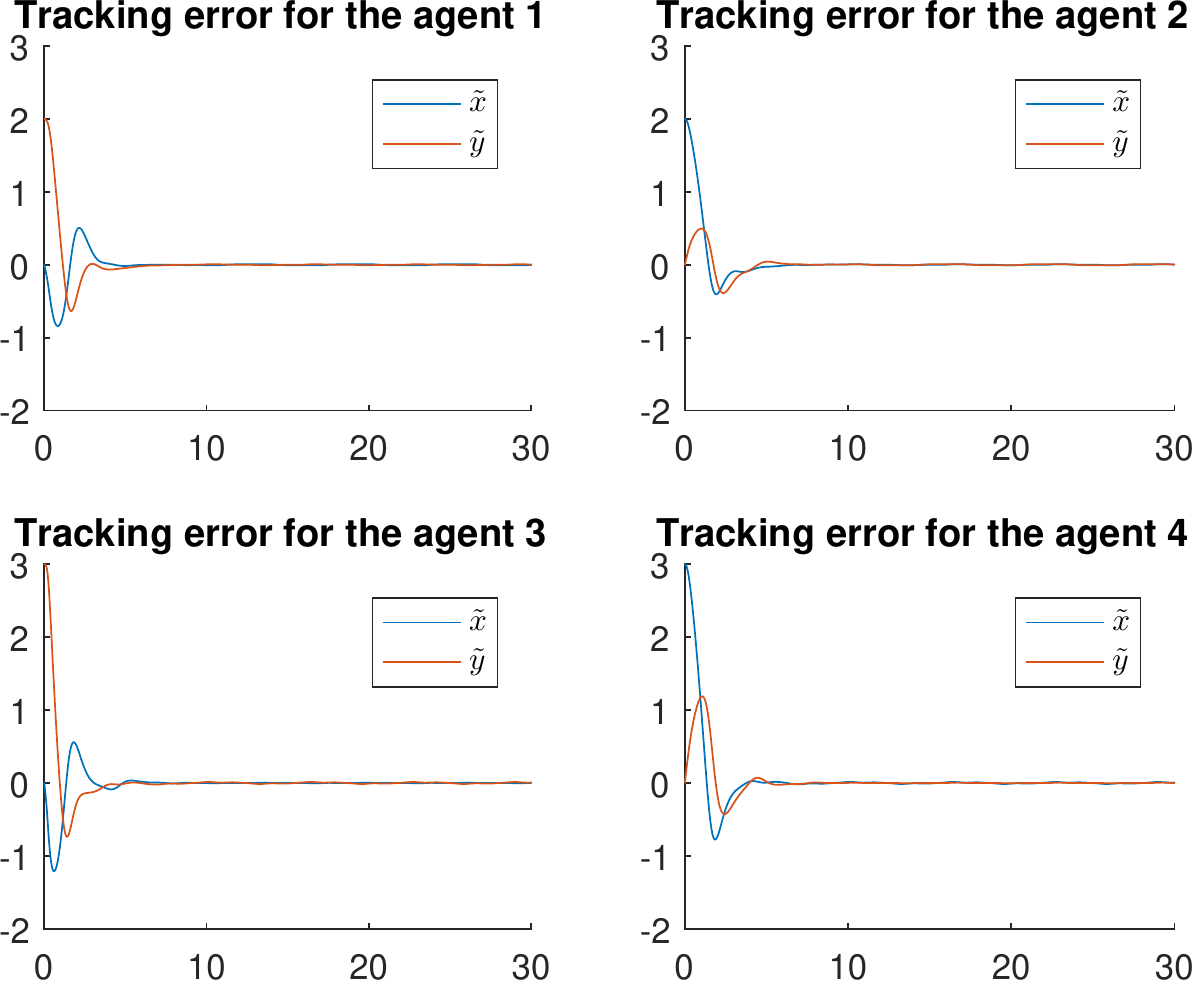}
		\caption{Tracking errors using controller (\ref{input_exp}) with (\ref{uc}).}
		\label{EXPC_tracking_error}
	\end{subfigure} \\
	\caption{Observer errors and tracking errors using observer-based controller.}
\end{figure*}

\begin{figure*} [htb!]
	\centering
	\begin{subfigure}{0.3\textwidth}
		\includegraphics[width=\linewidth]{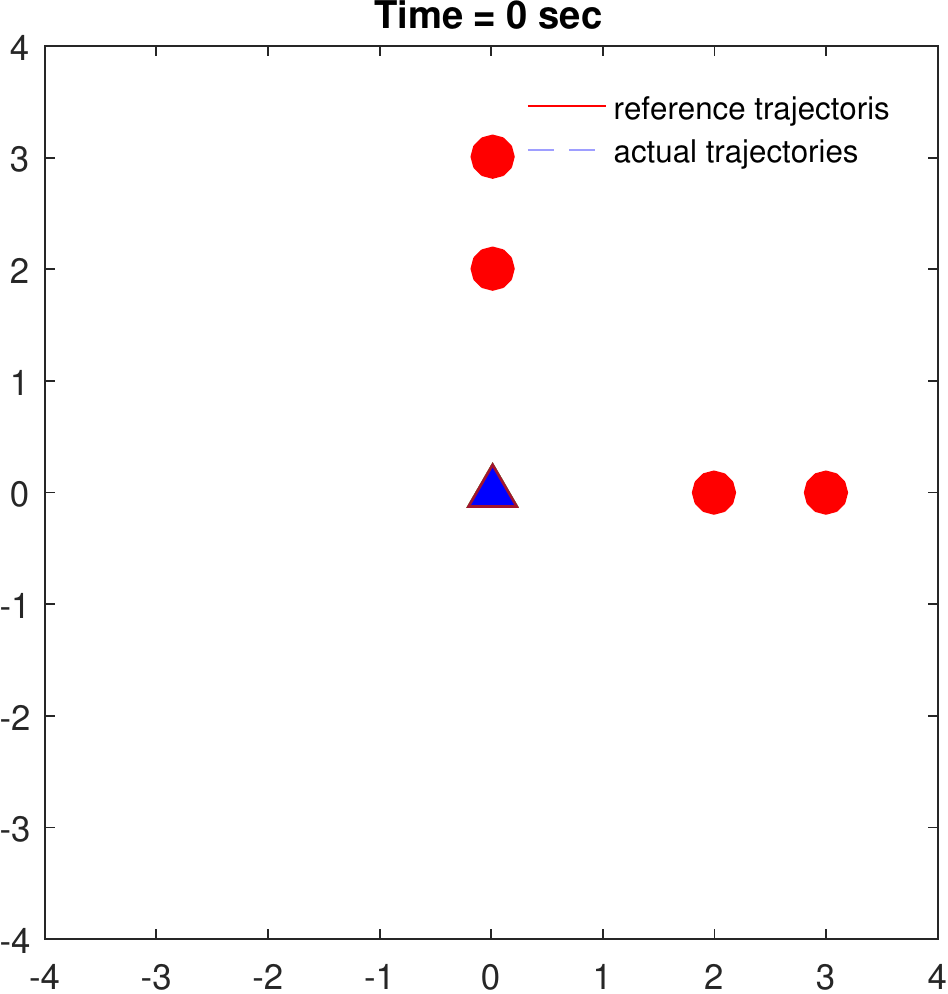}
		\caption{time at 0 seconds}
		\label{EXPC_0s}
	\end{subfigure} \quad
	\begin{subfigure}{0.3\textwidth}
		\includegraphics[width=\linewidth]{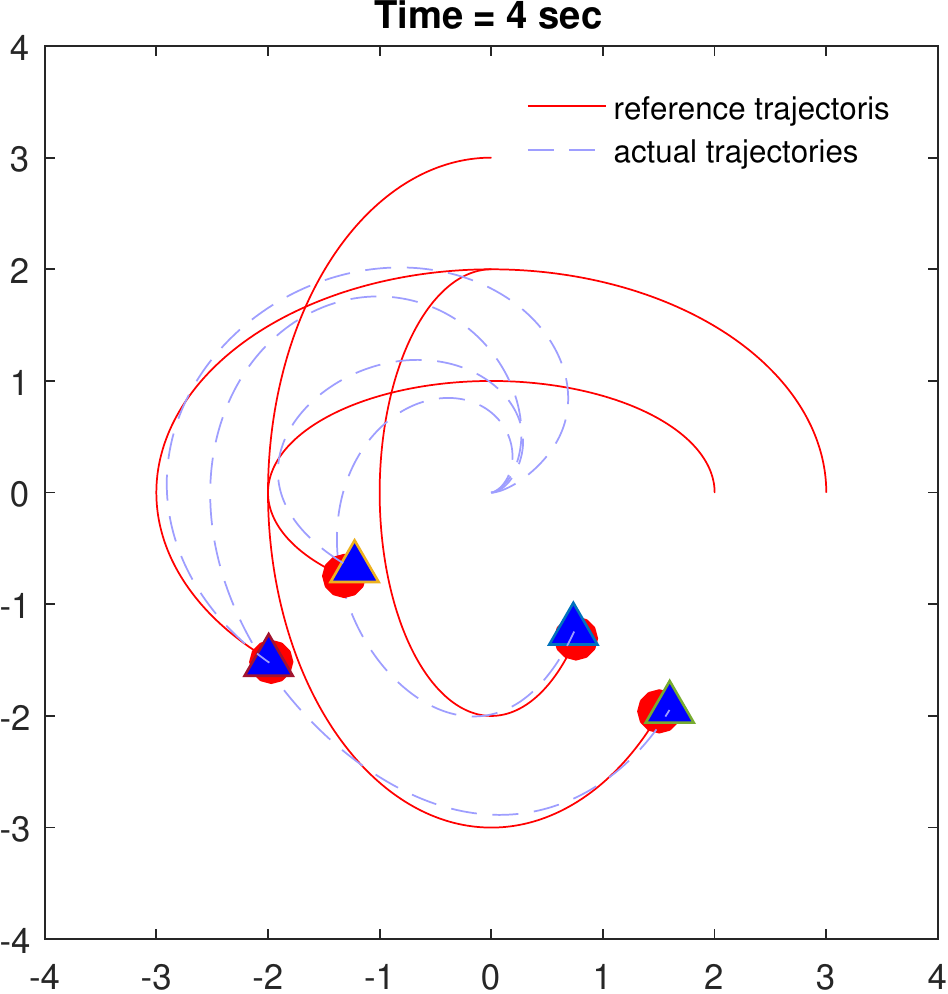}
		\caption{time at 4 seconds}
		\label{EXPC_4s}
	\end{subfigure} \quad
	\begin{subfigure}{0.3\textwidth}
		\includegraphics[width=\linewidth]{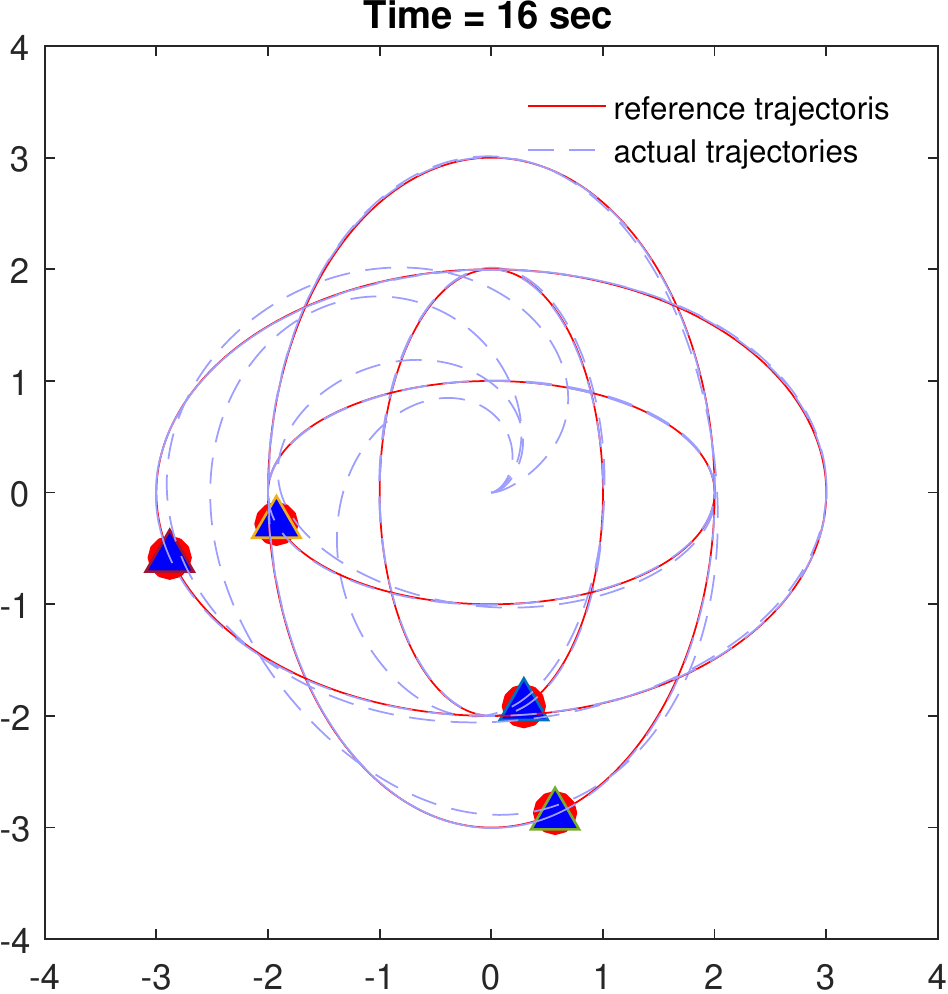}
		\caption{time at 16 seconds}
		\label{EXPC_16s}
	\end{subfigure} \\
	\caption{Snapshot of trajectory tracking using controller (\ref{input_exp}) with (\ref{uc}).}
	\label{2D_EXP}
\end{figure*}

To demonstrate the results of Theorem~\ref{Thm_experience}, which states that after the learning process, each vehicle is able to use the learned knowledge to follow any reference trajectory experienced by any vehicle on the learning stage.
In this part of our simulation, the experience-based controller (\ref{input_exp}) will be implemented with the same parameters as those of the previous subsection, such that vehicle 1 will follow the reference trajectory of vehicle 3, vehicle 2 will follow the reference trajectory of vehicle 1, and vehicle 3 will follow the reference trajectory of vehicle 2.
The initial position of the vehicles are set at the origin, with all velocities equal to zero.


%

Simulation results are shown as following.
Figure~\ref{EXPC_observer_error} shows that the observer error will converge to a close neighborhood around zero in a very short time period.
Figures~\ref{EXPC_0s} to \ref{EXPC_16s} show that all vehicles (blue triangles) will track its own reference trajectory (red solid circles), and Figure~\ref{EXPC_tracking_error} shows that all tracking errors $\tilde{x}_i$ and $\tilde{y}_i$ will converge to zero.

\section{Conclusion}

In this chapter, a high-gain observer-based CDL control algorithm has been proposed to estimate the unmodeled nonlinear dynamics of a group of homogeneous unicycle-type vehicles while tracking their reference trajectories.
It has been shown in this chapter that the state estimation, trajectory tracking, and consensus learning are all achieved using the proposed algorithm.
To be more specific, any vehicle in the system is able to learn the unmodeled dynamics along the union of trajectories experienced by all vehicles with the state variables provided by measurements and observer estimations.
In addition, we have also shown that with the converged NN weight, this knowledge can be applied on the vehicle to track any experienced trajectory with reduced computational complexity.
Simulation results have been provided to demonstrate the effectiveness of this proposed algorithm.


\end{document}